\documentclass[runningheads]{llncs}

\def\techreport{}

\usepackage{hyperref}
\usepackage{graphicx,subfigure}
\usepackage{amsmath,amssymb}
\usepackage{color}
\usepackage{tikz}
\usepackage{multirow}
\usepackage{wrapfig}
\usepackage{ifthen}

\usepackage{enumitem}
\setlist[itemize]{noitemsep, topsep=3pt}

\newcommand{\appref}[1]{Appx.~\ref{#1}}
\newcommand{\sectref}[1]{Sec.~\ref{#1}}

\newcommand{\figref}[1]{Fig.~\ref{#1}}
\newcommand{\tabref}[1]{Table~\ref{#1}}
\newcommand{\egref}[1]{Example~\ref{#1}}
\newcommand{\egegref}[2]{Examples~\ref{#1} and \ref{#2}}
\newcommand{\eqnref}[1]{(\ref{#1})}
\newcommand{\thmref}[1]{Thm.~\ref{#1}}
\newcommand{\propref}[1]{Proposition~\ref{#1}}

\newcommand{\defref}[1]{Defn.~\ref{#1}}

\newcounter{exampcount}
\newenvironment{examp}
{\refstepcounter{exampcount}
\vskip6pt\noindent
{\bf Example \arabic{exampcount}.}}
{}

\newcommand{\startpara}[1]{{%
\vskip6pt\noindent
{\bf #1.}}}

\def\ra{{\rightarrow}}

\def\cB{{\mathcal{B}}}

\def\cX{{\mathcal{X}}}

\def\Nset{\mathbb{N}}
\def\Qset{\mathbb{Q}}
\def\Rset{\mathbb{R}}

\def\Tset{\mathbb{T}}

\def\ra{\rightarrow} 
\def\rmdef{\,\stackrel{\mbox{\rm {\tiny def}}}{=}}
\newcommand{\dirac}[1]{\delta_{#1}} 
\newcommand{\sem}[1]{ [ \! [ {#1}  ]  \! ]} 
\newcommand{\true}{\mathtt{true}} 
\def\E{\mathbb{E}}
\renewcommand{\leq}{\leqslant}
\renewcommand{\geq}{\geqslant}
\renewcommand{\le}{\leqslant}

\def\squareforqed{\hbox{\rlap{$\sqcap$}$\sqcup$}}
\def\qed{\ifmmode\squareforqed\else{\unskip\nobreak\hfil
\penalty50\hskip1em\null\nobreak\hfil\squareforqed
\parfillskip=0pt\finalhyphendemerits=0\endgraf}\fi}

\newcommand\mdp{{\sf M}}
\newcommand\pomdp{{\sf M}}

\newcommand\pta{{\sf P}}
\newcommand\popta{{\sf P}}
\newcommand\sempopta{{\sem{\popta}}}

\newcommand\sinit{{\bar{s}}}
\newcommand\binit{{\bar{b}}}
\newcommand\oinit{{\bar{o}}}

\newcommand\act{\mathit{A}}

\newcommand\acts{\act}
\newcommand\dist{{\mathit{Dist}}}

\def\Pr{{\mathit{Pr}}}

\newcommand\strat{{\sigma}}
\newcommand\strats{{\Sigma}}

\newcommand{\preach}[3]{{\mathit{Pr}_{#1}^{#2}(\future #3)}}
\newcommand{\ereach}[3]{{\E_{#1}^{#2}(\future #3)}}

\newcommand{\estrat}[2]{\E_{#1}^{#2}} 
\newcommand{\rf}[2]{{\mathit{rew}^{#1}({#2})}} 

\newcommand{\ipaths}{\mathit{IPaths}}
\newcommand{\fpaths}{\mathit{FPaths}}
\newcommand{\ipat}{\omega}
\newcommand{\fpat}{\pi}

\renewcommand{\Pr}[2]{{\mathit{Pr}_{#1}^{#2}}}
\newcommand{\enab}[1]{A(#1)}

\newcommand{\obs}{\mathcal{O}}
\newcommand{\obsf}{\mathit{obs}}

\def\sat{{\,\models\,}}

\def\until{{\ {\cal U}\ }}

\def\tuntil{{\ {\cal U}^{\leq t}\ }}

\def\calPbp{{\mathtt P}_{\bowtie p}}

\newcommand{\calRbp}{{\mathtt R}_{\bowtie q}}
\def\calP{{\mathtt P}}
\def\calR{{\mathtt R}}

\def\until{{\, {\mathtt U}\, }}

\def\tuntil{{\, {\mathtt U}^{\leq t}\,}}
\def\future{{{\mathtt F}\,}}

\def\tfuture{{{\mathtt F}^{\leq t}\,}}

\newcommand{\scumul}[1]{\mathtt{C}^{#1}} 

\newcommand{\sinstant}[1]{\mathtt{I}^{#1}} 
\newcommand{\sreachrew}[1]{\mathtt{F}\,{#1}} 

\newcommand\probopP{{\mathtt P}}
\newcommand\probop[2]{\probopP_{#1}[\,{#2}\,]}
\newcommand\probopbp[1]{\probop{\bowtie p}{#1}}
\newcommand\rewopR{{\mathtt R}}
\newcommand\rewop[3]{\rewopR_{\,#2}[\,{#3}\,]}
\newcommand\rewopbr[2]{\rewop{#1}{\bowtie q}{#2}}

\newcommand{\ptatuple}{( \loc, \linit, \clocks, \act , \inv, \enb, \probt,\ptarew)}
\newcommand{\poptatuple}{( \loc, \linit, \clocks, \act , \inv, \enb, \probt,\ptarew, \obs_\loc,\obsf_\loc)}

\newcommand{\bptatuple}{( \dist^{\obsf_\loc}, \dirac{\linit}, \clocks, \act , \inv^\cB, \enb^\cB, \probt^\cB,\ptarew^\cB)}
\newcommand{\clocks}{\mathcal{X}}
\newcommand{\valuations}{\Rset^\clocks}

\newcommand{\tvaluations}{\Tset^\clocks}

\newcommand{\vinz}{\models}

\newcommand{\loc}{\mathit{L}}
\newcommand{\probt}{\mathit{prob}}

\newcommand{\linit}{\overline{l}}
\newcommand{\inv}{\mathit{inv}}
\newcommand{\enb}{\mathit{enab}}

\def\bk{{\mathbf{k}}}

\newcommand{\CC}[1]{\mathit{CC}({#1})}
\newcommand{\cc}{\zeta}

\newcommand{\ptarew}{{r}}
\newcommand{\ptaarew}{\ptarew_{\act}}
\newcommand{\lrew}{\ptarew_{L}} 

\definecolor{blue}{rgb}{0,0,1}

\begin{document}

\title{Verification and Control of Partially Observable Probabilistic Real-Time Systems}

\author{Gethin Norman\inst{1} \and David~Parker\inst{2} \and Xueyi Zou\inst{3}}

\institute{School of Computing Science, University of Glasgow, UK
\and School of Computer Science, University of Birmingham, UK
\and Department of Computer Science, University of York, UK}

\authorrunning{Norman \and Parker \and Zou}
\titlerunning{Partially Observable Probabilistic Real-Time Systems}

\maketitle

\begin{abstract}
We propose automated techniques for the verification and control of
probabilistic real-time systems that are only partially observable.
To formally model such systems, we define an extension of probabilistic timed automata
in which local states are partially visible to an observer or controller.
We give a probabilistic temporal logic that can express a range of quantitative properties of these models,
relating to the probability of an event's occurrence or the expected value of a reward measure.
We then propose techniques to either verify that such a property holds
or to synthesise a controller for the model which makes it true.
Our approach is based on an integer discretisation of the model's dense-time behaviour
and a grid-based abstraction of the uncountable belief space induced by partial observability.
The latter is necessarily approximate since the underlying problem is undecidable,
however we show how both lower and upper bounds on numerical results can be generated.
We illustrate the effectiveness of the approach by implementing it in the PRISM model checker
and applying it to several case studies, from the domains of computer security and task scheduling.
\end{abstract}


\section{Introduction}\label{intro-sec}

Guaranteeing the correctness of complex computerised systems
often needs to take into account quantitative aspects of system behaviour.
This includes the modelling of \emph{probabilistic} phenomena,
such as failure rates for physical components,
uncertainty arising from unreliable sensing of a continuous environment,
or the explicit use of randomisation to break symmetry.
It also includes \emph{real-time} characteristics,
such as time-outs or delays in communication or security protocols.
To further complicate matters, such systems are often \emph{nondeterministic}
because their behaviour depends on inputs or instructions from some external entity
such as a controller or scheduler.

Automated verification techniques such as probabilistic model checking
have been successfully used to analyse quantitative properties of probabilistic, real-time systems
across a variety of application domains, including
wireless communication protocols, computer security and task scheduling.
These systems are commonly modelled using \emph{Markov decision processes} (MDPs),
if assuming a discrete notion of time, or \emph{probabilistic timed automata} (PTAs),
if using a dense model of time.
On these models, we can consider two problems:
\emph{verification} that it satisfies some formally specified property
for any possible resolution of nondeterminism;
or, dually, \emph{synthesis} of a controller (i.e., a means to resolve nondeterminism)
under which a property is guaranteed. For either case, an important consideration is the extent to which the system's state
is \emph{observable} to the entity controlling it.
For example, to verify that a security protocol is functioning correctly, 
it may be essential to model the fact that some data held by a participant is not externally visible,
or, when synthesising a controller for a robot, the controller may not be implementable in practice
if it bases its decisions on information that cannot be physically observed.

Partially observable MDPs (POMDPs) are a natural way to extend MDPs in order to tackle this problem.
However, the analysis of POMDPs is considerably more difficult than MDPs
since key problems are undecidable~\cite{MHC03}.
A variety of verification problems have been studied for these models
(see e.g., \cite{dA99b,BBG08,CCT13}) and the use of POMDPs is common in
fields such as AI and planning~\cite{Cas98},
but there is limited progress in the development of practical
techniques for probabilistic verification in this area,
or exploration of their applicability.

In this paper, we present novel techniques for verification and control
of probabilistic real-time systems under partial observability.
We propose a model called \emph{partially observable probabilistic timed automata} (POPTAs),
which extends the existing model of PTAs
with a notion of partial observability. 
The semantics of a POPTA is an infinite-state POMDP.
We then
define a temporal logic, based on \cite{NPS13},
to express properties of POPTAs relating to the probability of an event
or the expected value of various reward measures.
Nondeterminism in a POPTA is resolved by a \emph{strategy}
that decides which actions to take and when to take them,
based only on the history of observations (not states).
The core problems we address are how to
\emph{verify} that a temporal logic property holds for all possible strategies,
and how to \emph{synthesise} a strategy under which the property holds. 

In order to achieve this, we use a combination of techniques.
First, we develop a \emph{digital clocks} discretisation for POPTAs,
which extends the existing notion for PTAs~\cite{KNPS06},
and reduces the analysis to a \emph{finite} POMDP.
We define the conditions under which properties in our temporal logic are
preserved and prove the correctness of the reduction.
To analyse the resulting POMDP, we use grid-based techniques~\cite{Lov91,Yu06}, 
which transform it to a fully observable but continuous-space MDP
and then approximate its solution based on a finite set of grid points.
We use this to construct and solve a strategy for the POMDP.
The result is a pair of lower and upper bounds on the property of interest for the original POPTA.
If the results are not precise enough, we can refine the grid and repeat.

We implemented these methods in a prototype tool based on PRISM~\cite{KNP11},
and investigated their applicability by developing three case studies:
a non-repudiation protocol, a task scheduling problem
and a covert channel prevention device (the NRL pump).
Despite the complexity of POMDP solution methods,
we show that useful results can be obtained, often with precise bounds.
In each case study, nondeterminism, probability, real-time behaviour \emph{and} partial observability
are all crucial ingredients to the analysis, a combination not supported
by any existing techniques or tools.

\startpara{Related work}
POMDPs are common in fields such as AI and planning, and have many applications~\cite{Cas98}.
They have also been studied in the verification community, e.g.~\cite{dA99b,BBG08,CCT13},
establishing undecidability and complexity results
for various qualitative and quantitative verification problems.
Work in this area often also studies related models such as
Rabin's probabilistic automata~\cite{BBG08}, which can be seen as a special case of POMDPs,
and partially observable stochastic games (POSGs)~\cite{CD14}, which generalise them.
More practically oriented work includes:~\cite{GR12},
which proposes a counterexample-driven refinement method
to approximately solve MDPs in which components have partial observability of each other;
and~\cite{CCH+11}, which synthesises concurrent program constructs,
using a search over memoryless strategies in a POSG.
Theoretical results~\cite{BDMP03} and algorithms~\cite{CDL+07,PF12b} have been
developed for synthesis of partially observable timed games.
In \cite{BDMP03}, it is shown that the synthesis problem
is undecidable and, if the resources of the controller are fixed, decidable but prohibitively expensive.
The algorithms require constraints on controllers:
in \cite{CDL+07}, controllers only respond to changes made by the environment
and, in \cite{PF12b}, their structure must be fixed in advance.
We are not aware of any work for probabilistic real-time models.
\vskip6pt
\noindent
\ifthenelse{\isundefined{\techreport}}{%
An extended version of this paper, with proofs, is available as \cite{tech}.
}{%
This paper is an extended version, with proofs, of \cite{NPZ15}.
}%

\section{Partially Observable Markov Decision Processes}\label{pomdps-sec}

We begin with background material on MDPs
and POMDPs. 
Let $\dist(X)$ denote the set of discrete probability distributions over a set $X$, $\dirac{x}$ the distribution that selects $x \in X$ with probability 1,
and $\Rset$ the set of non-negative real numbers.

\begin{definition}[MDP]
An MDP is a tuple $\mdp {=} (S,\sinit,A,P,R)$ where: $S$ is a set of states; $\sinit \in S$ an initial state; $A$ a set of actions;
$P : S {\times} A \rightarrow \dist(S)$ a (partial) probabilistic transition function;
and $R : S {\times} A \rightarrow \Rset$ a reward function.
\end{definition}
Each state $s$ of an MDP $\mdp$ has a set
$\enab{s}\rmdef \{a\in\acts \mid P(s,a) \text{ is defined}\}$ of \emph{enabled} actions. If action $a\in\enab{s}$ is selected, then the probability of moving to state $s'$ is $P(s,a)(s')$
and a reward of $R(s,a)$ is accumulated in doing so.
A \emph{path} of $\mdp$ is a finite or infinite sequence
$\ipat=s_0{a_0}s_1{a_1}\cdots$,
where $s_i\in S$, $a_i\in\enab{s_i}$ and $P(s_i,a_i)(s_{i+1}){>}0$ for all $i \in \Nset$.
We write $\fpaths_{\mdp}$ and $\ipaths_{\mdp}$, respectively,
for the set of all finite and infinite paths of $\mdp$ starting in the initial state $\sinit$.

A \emph{strategy} of $\mdp$ (also called a \emph{policy} or \emph{scheduler})
is a way of resolving the choice of action in each state,
based on the MDP's execution so far.
%
\begin{definition}[Strategy]
A \emph{strategy} of an MDP $\mdp{=}(S,\sinit,A,P,R)$ is a function $\strat:\fpaths_\mdp {\rightarrow} \dist(\acts)$
such that $\strat(s_0{a_0}s_1\dots s_n)(a){>}0$ only if $a\in\enab{s_n}$.
\end{definition}
A strategy is \emph{memoryless} if its choices only depend on the current state,
\emph{finite-memory} if it suffices to switch between a finite set of modes
and \emph{deterministic} if it always selects an action with probability 1.
The set of strategies of $\mdp$ is denoted by $\strats_\mdp$.

When $\mdp$ is under the control of a strategy $\strat$,
the resulting behaviour is captured by a probability measure $\Pr{\mdp}{\strat}$ over the infinite paths of $\mdp$~\cite{KSK76}.

\startpara{POMDPs}
POMDPs extend MDPs by restricting the extent to which their current state can be observed,
in particular by strategies that control them.
In this paper (as in, e.g., \cite{BBG08,CCT13}), we adopt the following notion of observability.
\begin{definition}[POMDP]\label{pomdp-def}
A POMDP is a tuple $\pomdp {=} (S,\sinit,A,P,R,\obs,\obsf)$ where:
$(S,\sinit,A,P,R)$ is an MDP;
$\obs$ is a finite set of \emph{observations};
and $\obsf : S \rightarrow \obs$ is a labelling of states with observations.
For any states $s,s'\in S$ with $\obsf(s){=}\obsf(s')$,
their enabled actions must be identical, i.e., $\enab{s}{=}\enab{s'}$.
\end{definition}
The current state $s$ of a POMDP cannot be directly determined,
only the corresponding observation $\obsf(s)\in\obs$.
More general notions of observations are sometime used,
e.g., that depend also on the previous action taken or are probabilistic.
Our analysis of probabilistic verification case studies
where partial observation is needed (see, e.g., \sectref{case-sec})
suggests that this simpler notion of observability will often suffice in practice.
To ease presentation, we assume the initial state is observable,
i.e., there exists $\oinit \in \obs$ such that $\obsf(s){=}\oinit$ if and only if $s{=}\sinit$.

The notions of paths, strategies and probability measures given above for MDPs transfer directly to POMDPs.
However, the set $\strats_\pomdp$ of all strategies for a POMDP $\pomdp$
only includes \emph{observation-based strategies},
that is, strategies $\strat$ such that, for any paths $\fpat=s_0{a_0}s_1\dots s_n$ and $\fpat'=s_0'{a_0}'s_1'\dots s_n'$
satisfying $\obsf(s_i)=\obsf(s_i')$ and $a_i=a_i'$ for all $i$, 
we have $\strat(\fpat)=\strat(\fpat')$.

Key properties for a POMDP (or MDP) are
the probability of reaching a target, and
the expected reward cumulated until this occurs.
Let $O$ denote the target (e.g., a set of observations of a POMDP).
Under a specific strategy $\strat$, 
we denote these two properties by $\preach{\pomdp}{\strat}{O}$
and $\ereach{\pomdp}{\strat}{O}$, respectively.

Usually, we are interested in the \emph{optimal} (minimum or maximum) values
$\smash{\preach{\pomdp}{opt}{O}}$ and $\smash{\ereach{\pomdp}{opt}{O}}$,
where $opt\in\{\min,\max\}$. For a MDP or POMDP $\pomdp$:
\[\begin{array}{rclcrcl}
\preach{\pomdp}{\min}{O} & \rmdef & \inf\nolimits_{\strat\in\strats_\pomdp} \preach{\pomdp}{\strat}{O}
& \hspace*{0.5cm} &
\ereach{\pomdp}{\min}{O} & \rmdef & \inf\nolimits_{\strat\in\strats_\pomdp}  \ereach{\pomdp}{\strat}{O}
\\
\preach{\pomdp}{\max}{O} & \rmdef & \sup\nolimits_{\strat\in\strats_\pomdp} \preach{\pomdp}{\strat}{O}
&&
\ereach{\pomdp}{\max}{O} & \rmdef & \sup\nolimits_{\strat\in\strats_\pomdp} \ereach{\pomdp}{\strat}{O}
\end{array}
\]

\startpara{Beliefs}
For POMDPs, determining the optimal probabilities and expected rewards defined above is
undecidable~\cite{MHC03}, making exact solution intractable.
A useful construction, e.g., as a basis of approximate solutions,
is the translation from a POMDP $\pomdp$ to a \emph{belief MDP} $\cB(\pomdp)$,
an equivalent (fully observable) MDP, whose (continuous) state space
comprises \emph{beliefs}, which are probability distributions over the state space of $\pomdp$.
Intuitively, although we may not know which of several observationally-equivalent states
we are currently in, we can determine the likelihood of being in each one,
based on the probabilistic behaviour of $\pomdp$.
\ifthenelse{\isundefined{\techreport}}{%
A formal definition is given below.
}{%
The formal definition is given below, and we include further details in  \appref{belief-appx}.
}%
\begin{definition}[Belief MDP]
Let $\pomdp{=}(S,\sinit,A,P,R,\obs,\obsf)$ be a POMDP.
The \emph{belief MDP\/} of $\pomdp$ is given by $\cB(\pomdp){=}(\dist(S),\dirac{\sinit},A,P^\cB,R^\cB)$
where, for any beliefs $b,b'\in\dist(S)$ and action $a\in A$:
\[
\begin{array}{rcl}
P^\cB(b,a)(b') & = & \mbox{$\sum_{s \in S}$} \; b(s) \cdot \left( \mbox{$\sum_{o \in \obs \wedge b^{a,o}=b'}$} \mbox{$\sum_{s' \in S \wedge \obsf(s')=o}$} \; P(s,a)(s') \right) \\
R^\cB(b,a) & = & \mbox{$\sum_{s \in S}$} \; R(s,a) \cdot b(s)
\end{array}
\]
and $b^{a,o}$ is the belief reached from $b$ by performing $a$ and observing $o$, i.e.: 
\[ 
 b^{a,o}(s') \; = \; \left\{ \begin{array}{cl}
\frac{\sum_{s \in S}  P(s,a)(s') \cdot b(s)}{\sum_{s \in S} b(s) \cdot \left( \sum_{s'' \in S \wedge \obsf(s'')=o} P(s,a)(s'') \right)} & \mbox{if $\obsf(s'){=}o$} \\
0 & \mbox{otherwise.}
\end{array} \right. 
\]
\end{definition}
The optimal values for the belief MDP equal those for the POMDP, e.g.\ we have:
\[
\preach{\pomdp}{\max}{O} = \preach{\cB(\pomdp)}{\max}{T_O} 
\; \; \mbox{and}
\; \; \ereach{\pomdp}{\max}{O} = \ereach{\cB(\pomdp)}{\max}{T_O}
\]
where $T_O = \{ b \in \dist(S) \, | \, \forall s \in S .\, (b(s){>}0 \ra \obsf(s)\in O) \}$.
\section{Partially Observable Probabilistic Timed Automata}\label{poptas-sec}

In this section, we define \emph{partially observable probabilistic timed automata} (POPTAs),
which generalise the existing model of probabilistic timed automata (PTAs) 
with the notion of partial observability from POMDPs explained in \sectref{pomdps-sec}.
We define the syntax of a POPTA, explain its semantics
(as an infinite-state POMDP) and define and discuss the \emph{digital clocks} semantics of a POPTA.



\startpara{Time \& clocks}
As in classical timed automata~\cite{AD94}, we model real-time behaviour using
non-negative, real-valued variables called \emph{clocks},
whose values increase at the same rate as real time.
Assuming a finite set of clocks $\clocks$,
a {\em clock valuation} $v$ is a function $v: \clocks {\ra} \Rset$
and we write $\valuations$ for the set of all clock valuations.
Clock valuations obtained from $v$ by incrementing all clocks by a delay $t \in \Rset$
and by resetting a set $X\subseteq\clocks$ of clocks to zero are denoted $v{+}t$ and $v[X{:=}0]$, respectively,
and we write ${\bf 0}$ if all clocks are 0.
A (closed, diagonal-free) {\em clock constraint\/} $\cc$ is either a conjunction of inequalities of the form
$x {\leq} c$ or $x {\geq} c$, where $x \in \clocks$ and $c \in \Nset$, or $\true$.
We write $v \vinz \cc$ if clock valuation $v$ satisfies clock constraint $\cc$
and use $\CC{\clocks}$ for the set of all clock constraints over $\clocks$.


\startpara{Syntax of POPTAs}
To explain the syntax of POPTAs, we first consider the simpler model of PTAs
and then show how it extends to POPTAs.
\begin{definition}[PTA syntax]\label{pta}
A PTA is a tuple $\pta {=} \ptatuple$ where:
\begin{itemize}
\item
$\loc$ is a finite set of \emph{locations}
and
$\linit\in\loc$ is an {\em initial location;}
\item
$\clocks$ is a finite set of {\em clocks}
and
$\act$ is a finite set of {\em actions;}
\item
$\inv: \loc \ra \CC{\clocks}$ is an {\em invariant condition;}
\item
$\enb: \loc {\times} \act \ra \CC{\clocks}$ is an {\em enabling condition;}
\item
$\probt : \loc {\times} \act \ra \dist(2^{\clocks} {\times} \loc)$ is a {\em probabilistic transition function;}
\item $\ptarew{=} (\lrew,\ptaarew)$ is a \emph{reward structure}
where $\lrew : \loc \rightarrow \Rset$ is a \emph{location reward function} and
$\ptaarew: \loc {\times} \act \ra \Rset$ is an \emph{action reward function}.
\end{itemize}
\end{definition}
A state of a PTA is a pair $(l,v)$ of location $l\in\loc$ and clock valuation $v\in\valuations$. Time $t\in \Rset$ can elapse in the state only if the invariant $\inv(l)$ remains continuously satisfied while time passes and the new state is then $(l,v{+}t)$. An action $a$ is enabled in the state if $v$ satisfies $\enb(l,a)$ and, if it is taken, then the PTA moves to location $l'$ and resets the clocks $X\subseteq\clocks$ with probability $\probt(l,a)(X,l')$.
PTAs have two kinds of rewards:
location rewards, which are accumulated at rate $\lrew(l)$ while in location $l$
and action rewards $\ptaarew(l,a)$, which are accumulated when taking action $a$ in location $l$. PTAs equipped with reward structures are a probabilistic extension of linearly-priced timed automata \cite{BFH+01b}.

\begin{definition}[POPTA syntax]\label{popta-def}
A partially observable PTA (POPTA) is a tuple $\popta = \poptatuple$ where:
\begin{itemize}
\item $\ptatuple$ is a \emph{PTA;}
\item $\obs_\loc$ is a finite set of \emph{observations;}
\item $\obsf_\loc : \loc \rightarrow \obs_\loc$ is a \emph{location observation function}.
\end{itemize}
For any locations $l,l' \in \loc$ with $\obsf_\loc(l){=}\obsf_\loc(l')$,
we require that $\inv(l){=}\inv(l')$
and $\enb(l,a){=}\enb(l',a)$ for all $a \in \act$.
\end{definition} %
The final condition ensures the semantics of a POPTA
yields a valid POMDP:
recall states with the same observation are required to have identical available actions.
Like for POMDPs, for simplicity,
we also assume that the initial location is observable,
i.e., there exists $\oinit \in \obs_\loc$ such that $\obsf_\loc(l){=}\oinit$ if and only if $l{=}\linit$.

The notion of observability for POPTAs is similar to the one for POMDPs, but applied to locations. Clocks, on the other hand, are always observable. The requirement that the same choices must be available in any observationally-equivalent states, implies the same delays must be available in observationally-equivalent states, and so unobservable clocks could not feature in invariant or enabling conditions. The inclusion of unobservable clocks would therefore necessitate modelling the system as a game with the elapse of time being under the control of a second (environment) player.  The underlying semantic model would then be a partially observable stochastic game (POSG), rather than a POMDP. However, unlike POMDPs, limited progress has been made on efficient computational techniques for this model
(belief space based techniques, for example, do not apply in general~\cite{CD14}). Even in the simpler case of non-probabilistic timed games,
allowing unobservable clocks requires algorithmic analysis to restrict
the class of strategies considered~\cite{CDL+07,PF12b}.

Encouragingly, however, we will later show in \sectref{case-sec}
that POPTAs with observable clocks were always sufficient for our modelling and analysis.

\startpara{Restrictions on POPTAs}
At this point, we need to highlight a few syntactic restrictions on the POPTAs treated in this paper.
Firstly, we emphasise that clock constraints appearing in a POPTA,
i.e., in its invariants and enabling conditions, are required to be
\emph{closed} (no strict inequalities) 
and \emph{diagonal-free} (no comparisons of clocks). 
This is a standard restriction when using the digital clocks discretisation~\cite{KNPS06}
which we work with in this paper.

Secondly, a specific (but minor) restriction for POPTAs is that
resets can only be applied to clocks that are non-zero.
The reasoning behind this is outlined later in \egref{counterexample-eg}.
Checking this restriction can easily be done when exploring the discrete
(digital clocks) semantics of the model -- see below and \sectref{mc-sec}.


\startpara{Semantics of POPTAs}
We now formally define the semantics of a POPTA $\popta$, which is given in terms of an infinite-state POMDP.
This extends the standard semantics of a PTA~\cite{NPS13} (as an infinite MDP)
with the same notion of observability we gave in \sectref{pomdps-sec} for POMDPs.
The semantics, $\sem{\popta}_\Tset$, is parameterised by a \emph{time domain} $\Tset$,
giving the possible values taken by clocks.
For the standard (dense-time) semantics of a POPTA, we take $\Tset=\Rset$.
Later, when we discretise the model, we will re-use this definition, taking $\Tset=\Nset$.
When referring to the ``standard'' semantics of $\popta$ we will often drop the subscript $\Rset$ and write $\sempopta$.

\begin{definition}[POPTA semantics]\label{poptasem-def}
Let $\popta{=}( \loc, \linit, \clocks, \act , \inv, \enb, \probt,\ptarew, \obs_\loc,$ $\obsf_\loc)$ be a POPTA.
The {\em semantics} of $\popta$, with respect to the time domain $\Tset$, 
is the POMDP $\sem{\popta}_\Tset{=}(S,\sinit,\act \cup \Tset,P,R,\obs_\loc {\times} \tvaluations, \obsf)$ such that:
\begin{itemize}
\item
$S = \{ (l,v) \in \loc {\times} \tvaluations \ | \ v \vinz \inv(l)\}$
and
$\sinit = (\linit,\mathbf{0})$;
\item
for $(l,v) \in S$ and $a \in \act \cup \Tset$, we have $P((l,v),a) = \mu$ if and only if:
\begin{itemize}
\item{(time transitions)}
$a \in \Tset$, $\mu = \dirac{(l,v + a)}$ and $v {+} t \vinz \inv(l)$ for all $0 {\leq} t {\leq} a$;
\item{(action transition)}
$a \in \act$, $v \vinz \enb(l,a)$ and for $(l',v') \in S$:
\[ \begin{array}{c}
\mu(l',v') = \sum_{X \subseteq \clocks \wedge v' = v[X{:=}0]}  \probt(l,a)(X,l')  
\end{array} \]
\end{itemize}
\vspace*{-0.75em}
\item for any $(l,v) \in S$ and $a \in \act \cup \Tset$, we have
$R((l,v),a) = \left\{ \begin{array}{cl}
\lrew(l){\cdot}a & \mbox{if $a \in \Tset$} \\
\ptaarew(l,a) & \mbox{if $a \in \act$}
\end{array} \right.$
\item
for any $(l,v) \in S$, we have $\obsf(l,v)=(\obsf_\loc(l),v)$.
\end{itemize}
\end{definition}


\begin{figure}[t]
\centering
\subfigure[]{\includegraphics[scale=0.39]{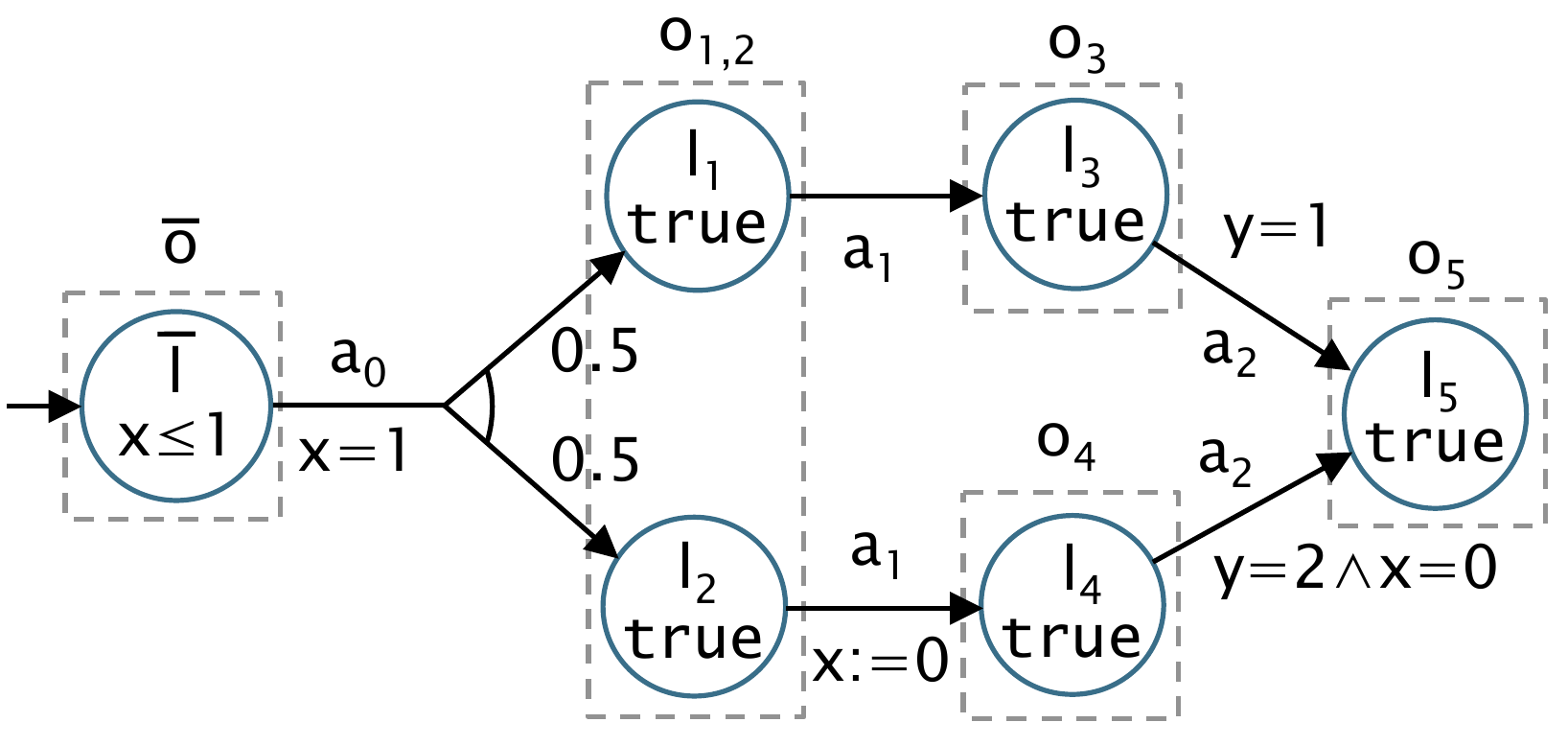}}
\hfill
\subfigure[]{\includegraphics[scale=0.39]{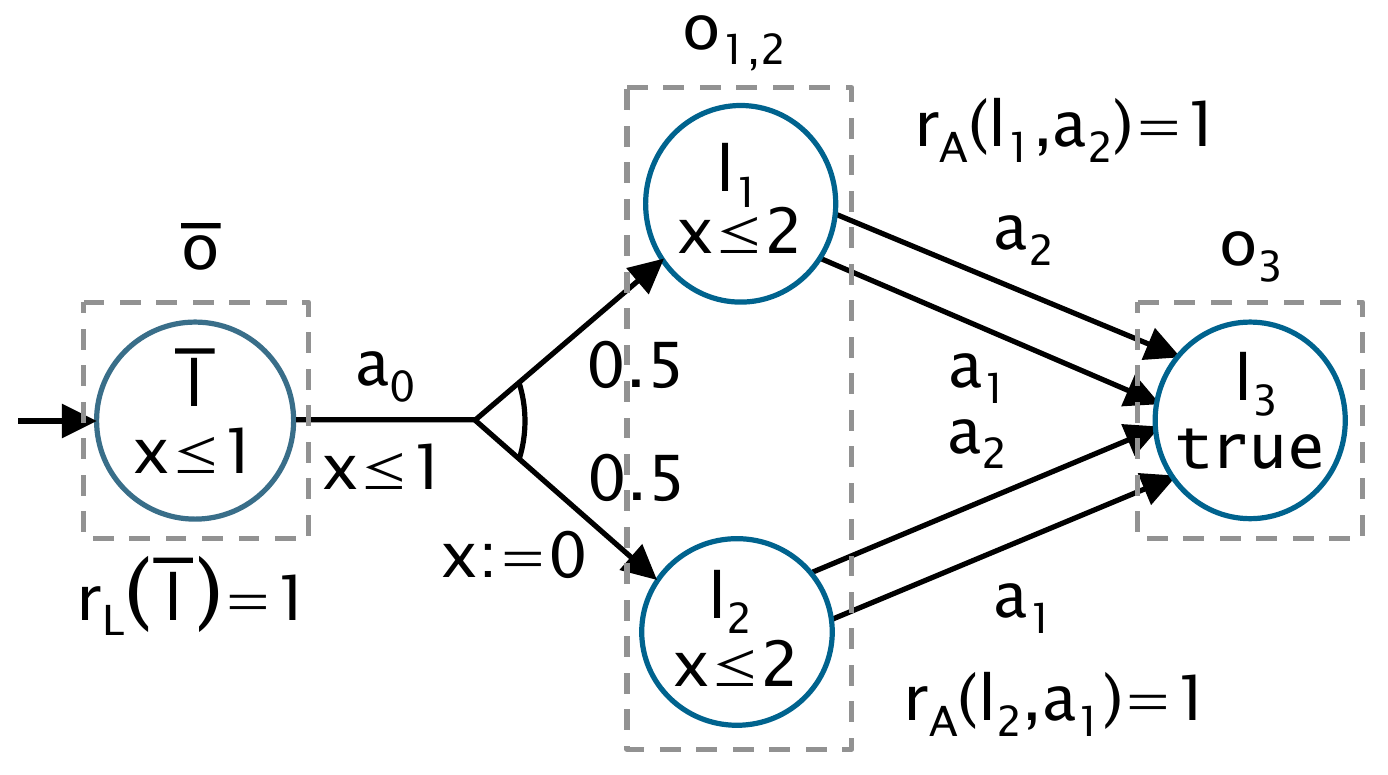}}
\vspace*{-0.4cm}
\caption{Examples of partially observable PTAs (see \egegref{example-eg}{counterexample-eg}).}\label{examples-fig}
\vspace*{-0.4cm}
\end{figure}

\begin{examp}\label{example-eg} 
Consider the POPTA in \figref{examples-fig}(a) with clocks $x,y$.
Locations are grouped according to their observations,
and we omit enabling conditions equal to $\mathtt{true}$.
We aim to maximise the probability of observing $o_5$.
If locations were fully observable, we would leave $\linit$ when $x{=}y{=}1$
and then, depending on whether the random choice resulted in a transition to $l_1$ or $l_2$,
wait $0$ or $1$ time units, respectively, before leaving the location.
This would allow us to move immediately from $l_3$ or $l_4$ to $l_5$,
meaning observation $o_5$ is seen with probability 1.
However, in the POPTA, we need to make the same choice in $l_1$ and $l_2$
since they yield the same observation. 
As a result, at most one of the transitions leaving locations $l_3$ and $l_4$ is enabled,
and the probability of observing $o_5$ is thus at most 0.5.
\end{examp}


\startpara{Digital clocks}
Since the semantics of a POPTA (like for a PTA) is an infinite-state model,
for algorithmic analysis, we first need to construct a \emph{finite} representation.
In this paper, we propose to use the \emph{digital clocks} approach,
generalising a technique already used for PTAs~\cite{KNPS06},
which in turn adapts one for timed automata~\cite{HMP92}.
In short, this approach discretises a POPTA model by transforming its real-valued clocks
to clocks taking values from a bounded set of integers.

For clock $x \in \cX$, let $\bk_x$ denote the greatest constant
to which $x$ is compared in the clock constraints of POPTA $\popta$.
If the value of $x$ exceeds $\bk_x$, its exact value will not affect the satisfaction of any invariants
or enabling conditions, and thus not affect the behaviour of $\popta$.
The digital clocks semantics, written $\sem{\popta}_\Nset$,
can be obtained from \defref{poptasem-def},
taking $\Tset$ to be $\Nset$ instead of $\Rset$.
We also need to redefine the operation $v{+}t$,
which now adds a delay $t\in\Nset$ to a clock valuation $v\in \Nset^{\cX}$:
we say that $v{+}t$ assigns the value $\min \{ v(x) {+} t , \bk_x{+}1 \}$ to each clock $x \in \clocks$.

%

Under the restrictions on POPTAs described above,
the digital semantics of a POPTA preserves the key properties required in this paper,
namely optimal probabilities and expected cumulative rewards
for reaching a specified observation.
\ifthenelse{\isundefined{\techreport}}{%
This is captured by the following theorem (the proof is available in \cite{tech}).}{%
This is captured by the following theorem, which we prove in \appref{correctness-appx}.
}%
\begin{theorem}\label{correctness-thm-new}
If $\popta$ is a closed, diagonal-free POPTA which resets only non-zero clocks,
then, for any set of observations $O$ of $\popta$ and $\mathit{opt} \in\{\min,\max\}$, we have:
\[
\preach{\sem{\popta}_{\Rset}}{opt}{O} = \preach{\sem{\popta}_\Nset}{opt}{O}
\hspace*{0.5em}\mbox{ and }\hspace*{0.6em}
\ereach{\sem{\popta}_{\Rset}}{opt}{O} = \ereach{\sem{\popta}_\Nset}{opt}{O}.
\]
\end{theorem}
The proof relies on showing probabilistic and expected reward values
agree on the belief MDPs underlying the POMDPs representing
the dense time and digital clocks semantics. This requires introducing the concept of a belief PTA for a POPTA (analogous to a belief MDP for a POMDP) and results for PTAs \cite{KNPS06}.

\begin{examp}\label{counterexample-eg} 
The POPTA $\popta$ in \figref{examples-fig}(b) demonstrates why our digital clocks approach (\thmref{correctness-thm-new})
is restricted to POPTAs which reset only non-zero clocks.
We aim to minimise the expected reward accumulated before observing $o_3$
(rewards are shown in \figref{examples-fig}(b) and are zero if omitted).
If locations were fully observable, the minimum reward would be 0,
achieved by leaving $\linit$ immediately
and then choosing $a_1$ in $l_1$ and $a_2$ in $l_2$. 
However, if we leave $\linit$ immediately, $l_1$ and $l_2$ are indistinguishable (we observe $(o_{1,2},(0))$ when arriving in both), 
so we must choose the same action in these locations, and hence the expected reward is 0.5.

Consider the strategy that waits $\varepsilon \in (0,1)$ before leaving $\linit$,
accumulating a reward of $\varepsilon$.
This is possible only in the dense-time semantics.
We then observe either $(o_{1,2},(\varepsilon))$ in $l_1$,
or $(o_{1,2},(0))$ in $l_2$. 
Thus, we see if $x$ was reset, determine if we are in $l_1$ or $l_2$, and take action $a_1$ or $a_2$ accordingly
such that no further reward is accumulated before seeing $o_3$, for a total reward of $\varepsilon$.
Since $\varepsilon$ can be arbitrarily small, the minimum (infimum) expected reward
for $\sem{\popta}_\Rset$ is 0. However, for the digital clocks semantics, we can only choose
a delay of $0$ or $1$ in $\linit$. For the former, the expected reward is 0.5, as described above;
for the latter, we can again pick $a_1$ or $a_2$ based on whether $x$ was reset, for a total expected reward of 1.
Hence the minimum expected reward for $\sem{\popta}_\Nset$ is 0.5, as opposed to 0 for  $\sem{\popta}_\Rset$.
\end{examp}

\section{Verification and Strategy Synthesis for POPTAs}\label{mc-sec}

We now present our approach for verification and strategy synthesis for POPTAs
using the digital clock semantics given in the previous section.

\startpara{Property specification}
First, we define a temporal logic for the formal specification of quantitative properties of POPTAs.
This is based on a subset (we omit temporal operator nesting) of the logic presented in~\cite{NPS13} for PTAs.
%
\begin{definition}[Properties]\label{logic-syntax-def}
The syntax of our logic is given by the grammar:
\[
\begin{array}{rclrcl}
\phi & ::= & \calPbp[\psi] \ \big|\ \calRbp [ \rho ] & 
\psi & ::= & \alpha\tuntil\alpha\ \big|\ \alpha\until\alpha \\[2pt]
\alpha & ::= & \true\ \big|\ o\ \big|\ \neg o\ \big|\ \cc \ \big|\ \alpha\wedge\alpha\ \big|\ \alpha\vee\alpha \hspace*{0.5cm} & 
\rho & ::= & \sinstant{=t} \ \big|\ \scumul{\le t}\ \big|\ \sreachrew{\alpha}
\end{array}
\]
where $o$ is an observation, $\cc$ is a clock constraint, ${\bowtie}\in\{{\leq},{<},{\geq},{>}\}$,
$p \in \Qset\cap[0,1]$, $q\in \Qset_{\geq 0}$ and $t \in\Nset$.
\end{definition}
A property $\phi$ is an instance of either the probabilistic operator $\calP$
or the expected reward operator $\calR$.
As for similar logics, $\calPbp[\psi]$ means the probability
of path formula $\psi$ being satisfied is ${\bowtie} p$,
and $\calRbp[\rho]$ the expected value of reward operator $\rho$ is ${\bowtie} q$.
For the probabilistic operator, we allow time-bounded ($\alpha\tuntil\alpha$) and unbounded ($\alpha\until\alpha$)
until formulas, and adopt the usual equivalences such as
$\sreachrew{\alpha} \equiv \true\until\alpha$ (``eventually $\alpha$'').
For the reward operator, we allow $\sinstant{=t}$ (location reward at time instant $t$),
$\scumul{\le t}$ (reward accumulated until time $t$) and $\sreachrew{\alpha}$
(the reward accumulated until $\alpha$ becomes true).
Our propositional formulas ($\alpha$) are Boolean combinations of observations and clock constraints.

We omit nesting of $\calP$ and $\calR$ operators for two reasons:
firstly, the digital clocks approach that we used to discretise time
is not applicable to nested properties (see~\cite{KNPS06} for details);
and secondly, it allows us to use a consistent property specification
for either verification or strategy synthesis problems
(the latter is considerably more difficult in the context of nested formulas~\cite{BGL04}). 
\begin{definition}[Property semantics] \label{logicsemantics-def}
Let $\popta$ be a POPTA with location observation function $\obsf_\loc$ and semantics $\sempopta$. We define satisfaction of a property $\phi$ from \defref{logic-syntax-def}
with respect to a strategy $\strat\in\strats_\sempopta$ as follows:
\[
\begin{array}{lll}
\sempopta,\strat \sat \probopbp{\psi} & \iff	& \Pr{\sempopta}{\strat}(\{\ipat\in\ipaths_{\sempopta}\ |\ \ipat\sat\psi\}) \bowtie p \\
\sempopta,\strat \sat \rewopbr{}{\rho}&\iff& \estrat{\sempopta}{\strat}(\rf{}{\rho})\bowtie q \\
\end{array}
\]
Satisfaction of a path formula $\psi$ by path $\ipat$, denoted $\ipat\models\psi$
and the random variable $\rf{}{\rho}$ for a reward operator $\rho$
are defined identically as for PTAs. Due to lack of space,
we omit their formal definition here and refer the reader to~\cite{NPS13}.
For a propositional formula $\alpha$ and state $s=(l,v)$ of $\sempopta$, we have
$s \sat o$ if and only if $\obsf_\loc(l){=}o$ and $s \sat \cc$ if and only if $v \vinz \cc$.
Boolean operators are standard.
\end{definition}


\startpara{Verification and strategy synthesis}
Given a POPTA $\popta$ and property $\phi$, we are interested in solving
the dual problems of \emph{verification} and \emph{strategy synthesis}.

\begin{definition}[Verification]
The \emph{verification} problem is: given a POPTA $\popta$ and property $\phi$,
decide if $\sempopta{,}\strat\sat\phi$ holds for all strategies $\strat{\in}\strats_{\sempopta}$.
\end{definition}
\begin{definition}[Strategy synthesis]
The \emph{strategy synthesis} problem is: given POPTA $\popta$ and property $\phi$,
find, if it exists, a strategy $\strat{\in}\strats_{\sempopta}$ such that $\sempopta{,}\strat\sat\phi$.
\end{definition}
The verification and strategy synthesis problems for $\phi$ can be solved similarly,
by computing \emph{optimal values} 
for either probability or expected reward objectives:
\[\begin{array}{rclcrcl}
\Pr{\sempopta}{\min}(\psi) & = & \inf\nolimits_{\strat\in\strats_\sempopta} \Pr{\sempopta}{\strat}(\psi)
& \hspace*{0.5cm} &
\estrat{\sempopta}{\min}(\rho) & = & \inf\nolimits_{\strat\in\strats_\sempopta} \estrat{\sempopta}{\strat}(\rho)
\\
\Pr{\sempopta}{\max}(\psi) & = & \sup\nolimits_{\strat\in\strats_\sempopta} \Pr{\sempopta}{\strat}(\psi) 
&&
\estrat{\sempopta}{\max}(\rho) & = & \sup\nolimits_{\strat\in\strats_\sempopta} \estrat{\sempopta}{\strat}(\rho) 
\end{array}
\] 
and, where required, also synthesising an \emph{optimal strategy}. 
For example, verifying $\phi {=} \probop{\geq p}{\psi}$
requires computation of $\Pr{\sempopta}{\min}(\psi)$ since $\phi$ is satisfied by all strategies
if and only if $\smash{\Pr{\sempopta}{\min}(\psi) {\geq} p}$.
Dually, consider synthesising a strategy for which $\phi' {=} \probop{\leq p}{\psi}$ holds.
Such a strategy exists if and only if $\smash{\Pr{\sempopta}{\min}(\psi) {\leq} p}$ and,
if it does, we can use the optimal strategy that achieves the minimum value.
A common practice in probabilistic verification
to simply query the optimal values directly,
using \emph{numerical} properties such as
\smash{$\probop{\mathtt{min=?}}{\psi}$ and $\rewop{r}{\mathtt{max=?}}{\rho}$}.

As mentioned earlier, when solving POPTAs (or POMDPs),
we may only be able to under- and over-approximate optimal values,
which requires adapting the processes sketched above.
For example, if we have determined lower and upper bounds
$\smash{p^\flat \leq \Pr{\sempopta}{\min}(\psi) \leq p^\sharp}$.
We can verify that $\phi {=} \probop{\geq p}{\psi}$ holds if $p^\flat \geq p$
or ascertain that $\phi$ does not hold if $p \geq p^\sharp$.
But, if $p^\flat < p < p^\sharp$, we need to refine our approximation
to produce tighter bounds.
An analogous process can be followed for the case of strategy synthesis. The remainder of this section therefore focuses on how to (approximately)
compute optimal values and strategies for POPTAs.



\startpara{Numerical computation algorithms}
Approximate numerical computation of either
optimal probabilities 
or expected reward values 
on a POPTA $\popta$ is performed with the sequence of steps given below,
each of which is described in more detail subsequently.
We compute both an under- and an over-approximation.
For the former, we also generate a strategy which achieves this value.
\begin{enumerate}\renewcommand{\labelenumi}{(\Alph{enumi})}
\item
We modify POPTA $\popta$, reducing the problem to computing optimal values for
a \emph{probabilistic reachability} or \emph{expected cumulative reward} property~\cite{NPS13};
\item
We apply the \emph{digital clocks} discretisation of \sectref{poptas-sec} 
to reduce the infinite-state semantics $\sem{\popta}_{\Rset}$ of $\popta$
to a \emph{finite-state POMDP} $\sem{\popta}_\Nset$;
\item
We build and solve a \emph{finite abstraction} of the (infinite-state) belief MDP $\cB(\sem{\popta}_\Nset)$
of the POMDP from (B), yielding an \emph{over-approximation}; 
\item
We synthesise and analyse a strategy for $\sem{\popta}_\Nset$, 
giving an \emph{under-approximation};
\item
If required, we \emph{refine} the abstraction's precision and repeat (C) and (D).
\end{enumerate}

\vspace*{-0.2em}
\startpara{(A) Property reduction}
As discussed in~\cite{NPS13} (for PTAs),
checking $\calP$ or $\calR$ properties of the logic of \defref{logic-syntax-def}
can always be reduced to checking either a probabilistic reachability ($\calPbp[\future \alpha]$)
or expected cumulative reward ($\calRbp[\future \alpha]$) property on a modified model.
For example, time-bounded probabilistic reachability ($\calPbp[\tfuture \alpha]$)
can be transformed into probabilistic reachability ($\calPbp[\future (\alpha \wedge y{\leq}t)]$)
where $y$ is a new clock added to the model.
We refer to \cite{NPS13} for full details.

\startpara{(B) Digital clocks}
We showed in \sectref{poptas-sec} that, assuming certain simple
restrictions on the POPTA $\popta$, we can construct a finite POMDP $\sem{\popta}_\Nset$
representing $\popta$ by treating clocks as bounded integer variables.
The translation itself is relatively straightforward,
involving a syntactic translation of the PTA (to convert clocks),
followed by a systematic exploration of its finite state space.
At this point, we also check satisfaction of the
restrictions on POPTAs described in \sectref{poptas-sec}.

\startpara{(C) Over-approximation}
We now solve the finite POMDP $\sem{\popta}_\Nset$.
For simplicity, here and below,
we describe the case of maximum reachability probabilities
(the other cases are very similar) and thus need to compute
$\smash{\preach{\sempopta_\Nset}{\max}{O}}$.
We first compute an \emph{over-approximation},
i.e. an \emph{upper} bound on the maximum probability.
This is computed from an approximate solution to the belief MDP $\cB(\sem{\popta}_\Nset)$,
whose construction we outlined in \sectref{pomdps-sec}.
This MDP has a continuous state space: the set of beliefs $\dist(S)$,
where $S$ is the state space of $\sem{\popta}_\Nset$.

To approximate its solution, we adopt the approach of \cite{Yu06} 
which computes values for a finite set of representative beliefs $G$ whose convex hull is $\dist(S)$.
Value iteration is applied to the belief MDP, using the computed values
for beliefs in $G$ and interpolating to get values for those not in $G$.
The resulting values give the required upper bound.
We use~\cite{Yu06} as it works with \emph{unbounded} (infinite horizon)
and \emph{undiscounted} properties.
There are many other similar approaches~\cite{SPK13},
but these are formulated for discounted or finite-horizon properties.

The representative beliefs can be chosen in a variety of ways.
We follow \cite{Lov91}, where $\smash{G = \{ \frac{1}{M} v  \, | \, v \in \Nset^{|S|} \wedge \sum_{i=1}^{|S|} v(i) {=} M \}}$, i.e. a uniform \emph{grid} with \emph{resolution} $M$.
A benefit is that interpolation is very efficient, using a process called triangulation~\cite{Eav84}.
A downside is that the grid size is exponential $M$.

\startpara{(D) Under-approximation}
Since it is preferable to have two-sided bounds, we also compute an \emph{under-approximation}:
here, a lower bound on $\smash{\preach{\sempopta_\Nset}{\max}{O}}$.
To do so, we first synthesise a finite-memory strategy $\sigma^*$ for $\sem{\popta}_\Nset$
(which is often a required output anyway).
The choices of this strategy are built by stepping through the belief MDP and, for the current belief, choosing an action that achieves the values returned by value iteration in (C) above -- see for example~\cite{SPK13}.
We then compute, by building and solving the finite Markov chain induced by
$\sem{\popta}_\Nset$ and $\sigma^*$, the value
$\preach{\sempopta_\Nset}{\strat^*}{O}$ which is a lower bound  for $\smash{\preach{\sempopta_\Nset}{\max}{O}}$.

\startpara{(E) Refinement}
Finally, although no a priori bound can be given on the error
between the generated under- and over-approximations (recall that the basic problem is undecidable),
asymptotic convergence of the grid based approach \emph{is} guaranteed \cite{Yu06}.
In practice, if the computed approximations do not suffice to verify the required property
(or, for strategy synthesis, $\strat^*$ does not satisfy the property),
then we increase the grid resolution $M$ and repeat steps (C) and (D).

\section{Implementation and Case Studies}\label{case-sec}

We have built a prototype tool for verification and strategy synthesis of POPTAs and POMDPs
as an extension of PRISM~\cite{KNP11}. 
We extended the existing modelling language for PTAs,
to allow model variables to be specified as observable or hidden.
The tool performs the steps outlined in \sectref{mc-sec},
computing a pair of bounds for a given property and synthesising a corresponding strategy.
We focus on POPTAs, but the tool can also analyse POMDPs directly.
The software, details of all case studies, parameters
and properties are available online at: 
\vspace*{-0.2em}
\begin{center}\url{http://www.prismmodelchecker.org/files/formats15poptas/} \end{center}
\vspace*{-0.2em}\noindent
We have developed three case studies to evaluate the tool and techniques, 
discussed in more detail below.
In each case, nondeterminism, probability, real-time behaviour \emph{and} partial observability
are all essential aspects required for analysis.

\startpara{The NRL pump}
The NRL (Naval Research Laboratory) pump~\cite{KMM98} is designed to 
provide reliable and secure communication over networks of nodes
with `high' and `low' security levels.
It prevents a covert channel leaking information
from `high' to `low' through the \emph{timing} of messages and acknowledgements. Communication is buffered and \emph{probabilistic} delays are added to acknowledgements from `high'
in such a way that the potential for information leakage is minimised,
while maintaining network performance. A PTA model is considered in \cite{LMT+04}.

We model the pump as a POPTA using a hidden variable for a secret value	$z\in\{0,1\}$
(initially set uniformly at random)
which `high' tries to covertly communicate to `low'.
This communication is attempted by adding a delay of $h_0$ or $h_1$,
depending on the value of $z$, whenever sending an acknowledgement to `low'.
In the model, `low' sends $N$ messages to `high' and tries to guess $z$
based on the time taken for its messages to be acknowledged.
We consider the maximum probability `low' can (either eventually or within some time frame) correctly guess $z$. We also study the expected time to send all messages and acknowledgements.
These properties measure the security and performance aspects of the pump.
Results are presented in \figref{pump1-fig} varying $h_1$ and $N$ (we fix $h_0{=}2$).
They show that increasing either the difference between $h_0$ and $h_1$ (i.e., increasing $h_1$)
or the number $N$ of messages sent improve the chance of `low' correctly guessing the secret $z$, at the cost of a decrease in network performance.
On the other hand, when $h_0{=}h_1$, however many messages are sent,
`low', as expected, learns nothing of the value being sent
and at best can guess correctly with probability 0.5.

\begin{figure}[t]
\vspace*{-0.1cm}
\centering
\subfigure[]{\includegraphics[scale=0.27]{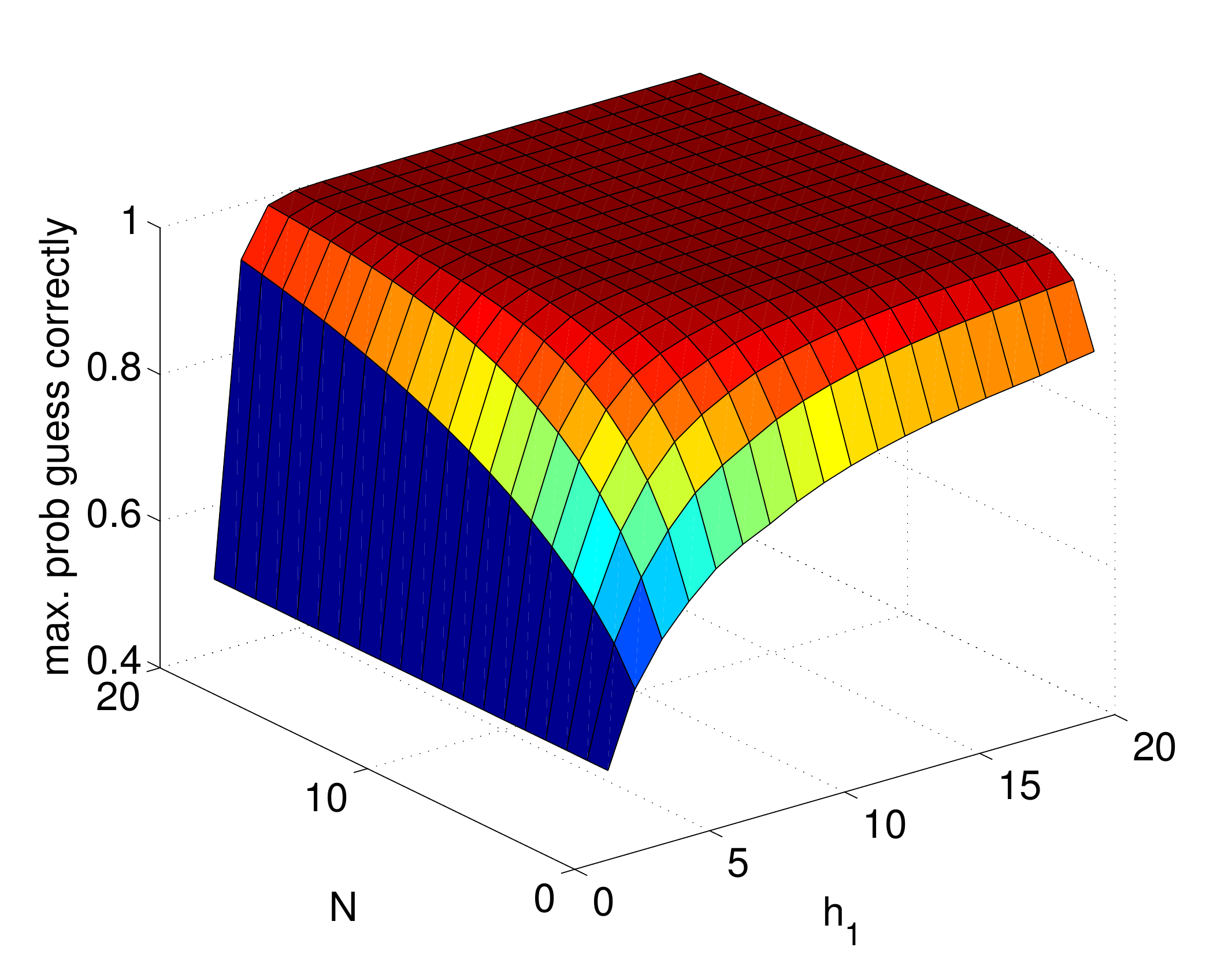}}
\hfill
\subfigure[]{\includegraphics[scale=0.27]{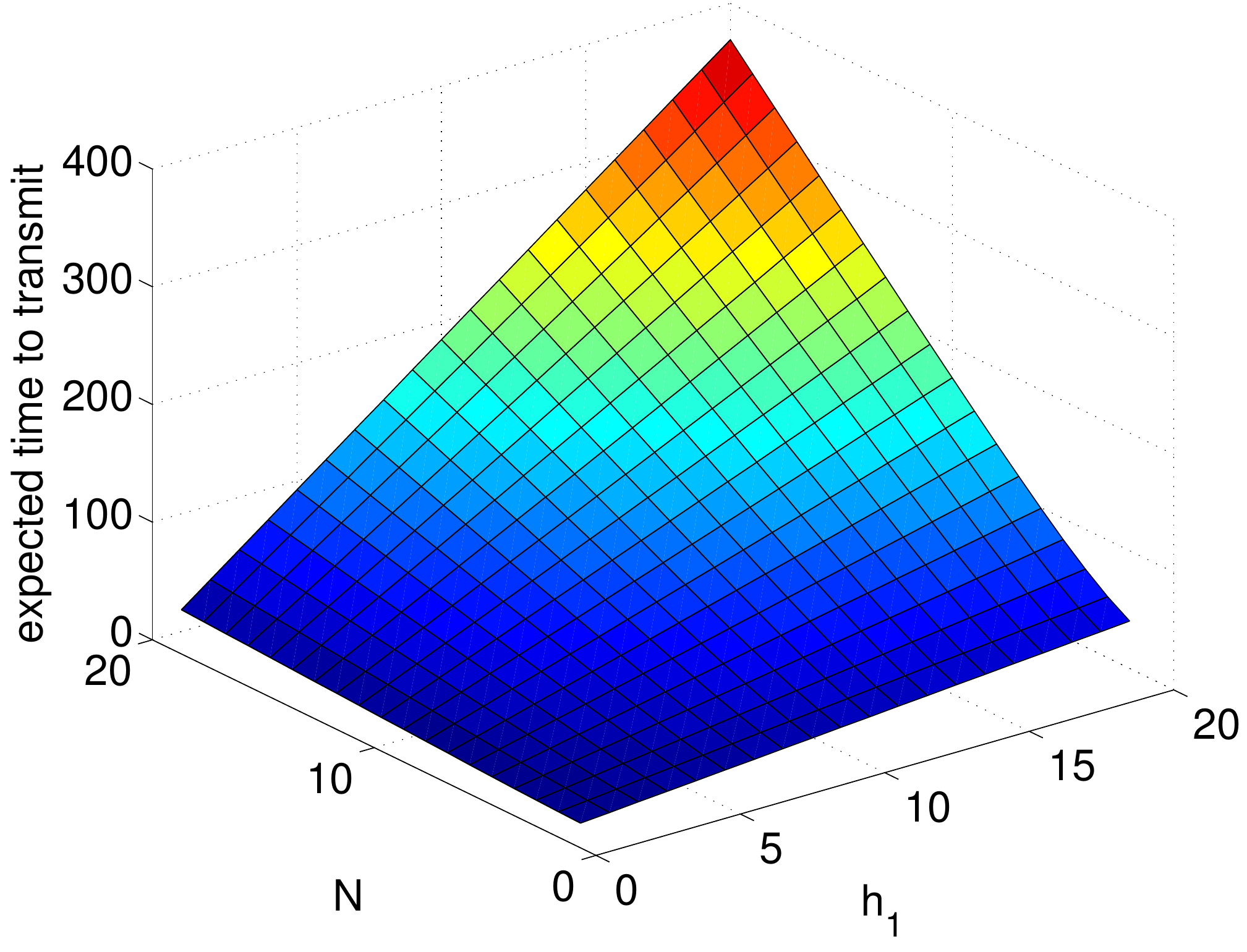}}
\vspace*{-0.3cm}
\caption{Analysing security/performance of the NRL pump:
(a) Maximum probability of covert channel success;
(b) Maximum expected transmission time.}\label{pump1-fig}
\vspace*{-0.4cm}
\end{figure}

\startpara{Task-graph scheduler} 
Secondly, we consider a task-graph scheduling problem adapted from \cite{BFLM11},
where the goal is to minimise the \emph{time} or \emph{energy consumption} required to evaluate
an arithmetic expression on multiple processors with different speeds and energy consumption.
We extend both the basic model of \cite{BFLM11} and the extension from \cite{NPS13}
which uses PTAs to model \emph{probabilistic} task execution times.
A new `low power' state to one processor, 
allowing it to save energy when not in use,
but which incurs a delay when waking up to execute a new task.
This state is entered with probability $\mathit{sleep}$ after each task is completed.
We assume that the scheduler cannot observe whether the processor enters this lower power state, and hence the model is a POPTA.
We generate optimal schedulers (minimising expected execution time or energy usage) using strategy synthesis.

\startpara{Non-repudiation protocol}
Our third case study is a non-repudiation protocol for information transfer
due to Markowitch \& Roggeman~\cite{MR99}.
It is designed to allow an originator $O$ to send information
to a recipient $R$ while guaranteeing \emph{non-repudiation},
that is, neither party can deny having participated in the information transfer.
The initialisation step of the protocol requires $O$ to \emph{randomly} select an integer $N$
in the range $1,\dots,K$
that is never revealed to $R$ during execution. In previous analyses~\cite{LMT05,NPS13},
modelling this step was not possible since no notion of (non-)observability was used.
We resolve this by building a POPTA model of the protocol
including this step, thus matching Markowitch \& Roggeman's original specification.
In particular, we include a hidden variable to store the random value $N$.
We build two models:
a basic one, where $R$'s only malicious behaviour corresponds to stopping early;
and a second, more complex model, where $R$ has access to a decoder. 
We compute the maximum probability that $R$ gains an unfair advantage
(gains the information from $O$ while being able to deny participating).
Our results (see \tabref{table-tab}) show that, for the basic model, this probability equals $1/K$ and $R$ is more powerful in the complex
model.

\startpara{Experimental results}
\tabref{table-tab} summarises a representative set of experimental results
from the analysis of our three case studies. 
All were run on a 2.8 GHz PC with 8GB RAM. 
The table shows the parameters used for each model
(see the web page cited above for details),
the property analysed and various statistics from the analysis:
the size of the POMDP obtained through the digital clocks semantics; number of observations;
number of hidden values (i.e., the maximum number of states with the same observation);
the grid size (resolution $M$ and total number of points); the time taken; and the results obtained.
For comparison, in the rightmost column, we show what result is obtained if the
POPTA is treated as a PTA (by making everything observable).

\begin{table}[!t]
\centering
{\scriptsize \hspace*{-0.1cm}
\begin{tabular}{|c|c|c||r|r|r|r|r|r|c||r|}
\hline
\multicolumn{2}{|c|}{\multirow{3}{2.3cm}{\centering Case study \\ (parameters)}} & \multicolumn{1}{|c||}{\multirow{3}{1.3cm}{\centering Property}}  & \multicolumn{7}{|c||}{Verification/strategy synthesis of POPTA} & \multicolumn{1}{|c|}{\multirow{3}{0.7cm}{\centering PTA \\ result}} \\
\cline{4-10}
\multicolumn{2}{|c|}{} & & \multicolumn{1}{|c|}{States} &\multicolumn{1}{|c|}{Num.} &\multicolumn{1}{|c|}{Num.} & \multicolumn{1}{|c|}{Res.} & \multicolumn{1}{|c|}{Grid} & \multicolumn{1}{|c|}{Time} & \multicolumn{1}{|c||}{Result} & \\ 
\multicolumn{2}{|c|}{} & & \multicolumn{1}{|c|}{($\sem{\popta}_\Nset$)} & \multicolumn{1}{|c|}{obs.} & \multicolumn{1}{|c|}{hidd.} & \multicolumn{1}{|c|}{($M$)} & \multicolumn{1}{|c|}{points} & \multicolumn{1}{|c|}{(s)} & \multicolumn{1}{|c||}{(bounds)} & \\
\hline \hline
\multirow{3}{1.3cm}{\centering \emph{pump} \\ (\emph{$h_1$ $N$})}
& $16$ $2$ & \multirow{3}{2.1cm}{\centering $\mathtt{P_{max=?}}[\mathtt{F} \; \mathit{guess}]$} & 243 & 145 & 3 & 2 & 342 & 0.7 & \mbox{$[0.940,0.992]$} & 1.0 \\ 
& $16$ $2$ & & 243 & 145 & 3 & 40 &  4,845 & 4.0 & \mbox{$[0.940,0.941]$} & 1.0 \\ 
& $16$ $16$ & & 1,559 & 803 & 3  & 2 &  2,316 & 16.8 & \mbox{$[0.999,0.999]$} & 1.0 \\ 
\hline
\multirow{4}{1.3cm}{\centering \emph{pump} \\(\emph{$h_1$ $N$ $D$})}
& $8$ $4$ $50$ & \multirow{4}{2.1cm}{\centering $\mathtt{P_{max=?}}[\mathtt{F}^{\leq D} \mathit{guess}]$} & 12,167 & 7,079 & 3 & 2 & 17,256 & 11.0 & \mbox{$[0.753,0.808]$} & 1.0 \\ 
& $8$ $4$ $50$ & & 12,167 & 7,079 & 3 & 12 & 68,201 & 36.2 & \mbox{$[0.763,0.764]$} & 1.0 \\ 
& $16$ $8$ $50$ & & 26,019 & 13,909 & 3 & 2 & 38,130 & 52.8 & \mbox{$[0.501,0.501]$} & 1.0 \\ 
& $16$ $8$ $100$ & & 59,287 & 31,743 & 3 & 2 & 86,832 & 284.8 & \mbox{$[0.531,0.532]$} & 1.0 \\ 
\hline
\hline
\multirow{3}{1.3cm}{\centering \emph{scheduler} \\ \emph{basic} \\ ($\mathit{sleep}$) }
& $0.25$ & \multirow{3}{2.1cm}{\centering $\mathtt{R_{min=?}}[\mathtt{F} \; \mathit{done}]$ \\ (exec. time)} & 5,002 & 3,557 & 2 & 2 & 6,447 & 3.2  & \mbox{$[14.69,14.69]$} & 14.44 \\ 
& $0.5$ & & 5,002 & 3,557 & 2 & 2 & 6,447 & 3.1 & \mbox{$[17.0,17.0]$} & 16.5 \\ 
& $0.75$ & & 5,002 & 3,557 & 2 & 4 & 9,337 & 3.1 & \mbox{$[19.25,19.25]$} & 18.5 \\ 
\hline
\multirow{3}{1.3cm}{\centering \emph{scheduler} \\ \emph{basic} \\ ($\mathit{sleep}$) }
& $0.25$ & \multirow{3}{2.1cm}{\centering $\mathtt{R_{min=?}}[\mathtt{F} \; \mathit{done}]$ \\ (energy cons.)} & 5,002 & 3,557 & 2 & 4 & 9,337 & 3.1 & \mbox{$[1.335,1.335]$} & 1.237 \\ 
& $0.5$ & & 5,002 & 3,557 & 2 & 2 & 6,447 & 3.1  & \mbox{$[1.270,1.270]$} & 1.186 \\ 
& $0.75$ & & 5,002 & 3,557 & 2 & 2 & 6,447 & 3.2 & \mbox{$[1.204,1.204]$} & 1.155 \\ 
\hline
\multirow{3}{1.3cm}{\centering \emph{scheduler} \\ \emph{prob} \\ ($\mathit{sleep}$) }
& $0.25$ & \multirow{3}{2.1cm}{\centering $\mathtt{R_{min=?}}[\mathtt{F} \; \mathit{done}]$ \\ (exec. time)} & 6,987 & 5,381 & 2 & 2 & 8,593 & 5.8 & \mbox{$[15.00,15.00]$} & 14.75 \\ 
& $0.5$ & & 6,987 & 5,381 & 2 & 2 & 8,593 & 5.8 & \mbox{$[17.27,17.27]$} & 16.77 \\ 
& $0.75$ & & 6,987 & 5,381 & 2 & 4 & 11,805 & 5.0 & \mbox{$[19.52,19.52]$} & 18.77 \\ 
\hline
\multirow{3}{1.3cm}{\centering \emph{scheduler} \\ \emph{prob} \\ ($\mathit{sleep}$) }
& $0.25$ & \multirow{3}{2.1cm}{\centering $\mathtt{R_{min=?}}[\mathtt{F} \; \mathit{done}]$ \\ (energy cons.) } & 6,987 & 5,381 & 2 & 4 & 11,805 & 5.3 & \mbox{$[1.335,1.335]$} & 1.3 \\ 
& $0.5$ & & 6,987 & 5,381 & 2 & 2 & 8,593 & 5.0 & \mbox{$[1.269,1.269]$} & 1.185 \\ 
& $0.75$ & & 6,987 & 5,381 & 2 & 2 & 8,593 & 5.8 & \mbox{$[1.204,1.204]$} & 1.155 \\ 
\hline
\hline
\multirow{4}{1.3cm}{\centering \emph{nrp} \\ \emph{basic} \\ $(K)$}
& 4 &  \multirow{4}{2.1cm}{\centering $\mathtt{P_{max=?}}[\mathtt{F} \; \mathit{unfair}]$} & 365 & 194 & 5 & 8 & 5,734 & 0.8 & \mbox{$[0.25,0.281]$} & 1.0 \\ 
& 4 &  & 365 & 194 & 5 & 24 & 79,278 & 5.9  & \mbox{$[0.25,0.25]$} & 1.0 \\ 
& 8 &  & 1,273 & 398 & 9 & 4 & 23,435 & 4.8 & \mbox{$[0.125,0.375]$} & 1.0 \\ 
& 8 &  & 1,273  & 398 & 9 & 8 & 318,312 & 304.6 & \mbox{$[0.125,0.237]$} & 1.0 \\ 
\hline
\multirow{4}{1.3cm}{\centering \emph{nrp} \\ \emph{complex} \\ $(K)$}
& 4 &  \multirow{4}{2.1cm}{\centering $\mathtt{P_{max=?}}[\mathtt{F} \; \mathit{unfair}]$} & 1,501 & 718 & 5 & 4 & 7,480 & 2.1 & \mbox{$[0.438,0.519]$} & 1.0 \\ 
& 4 &  & 1,501 & 718 & 5 & 12 & 72,748 & 14.8 & \mbox{$[0.438,0.438]$} & 1.0 \\ 
& 8 &  & 5,113 & 1,438 & 9 & 2 & 16,117 & 6.1 & \mbox{$[0.344,0.625]$} & 1.0 \\ 
& 8 &  & 5,113 & 1,438 & 9 & 4 & 103,939 & 47.1 & \mbox{$[0.344,0.520]$} & 1.0 \\ 
\hline
\end{tabular}}
\vspace*{0.1cm}
\caption{Experimental results from verification/strategy synthesis of POPTAs.}\label{table-tab}
\vspace{-0.8cm}
\end{table}

On the whole, we find that the performance of our prototype is good,
especially considering the complexity of the POMDP solution methods
and the fact that we use a relatively simple grid mechanism.
We are able to analyse POPTAs whose integer semantics yields POMDPs of up to 60,000 states,
with experiments usually taking just a few seconds and, at worst, 5-6 minutes. These are, of course, smaller than the standard PTA (or MDP) models that can be verified,
but we were still able to obtain useful results for several case studies.

The values in the rightmost column of \tabref{table-tab} illustrate that the results obtained with POPTAs would not have been possible using a PTA model, i.e.,
where all states of the model are observable.
For the $\mathit{pump}$ example, the PTA gives probability 1 of guessing correctly
(`low' can simply read the value of the secret).
For the $\mathit{scheduler}$ example, the PTA model gives a scheduler
with better time/energy consumption but that cannot be implemented in practice
since the power state is not visible. For the \emph{nrp} models, the PTA gives probability 1 of unfairness as the recipient can read the random value the originator selects.

Another positive aspect is that, in many cases, the bounds generated
are very close (or even equal, in which case the results are exact).
For the $\mathit{pump}$ and $\mathit{scheduler}$ case studies,
we included results for the smallest grid resolution $M$ required
to ensure the difference between the bounds is at most $0.001$.
In many cases, this is achieved with relatively small values
(for the $\mathit{scheduler}$ example, in particular, $M$ is at most 4).
For $\mathit{nrp}$ models, we were unable to do this when $K{=}8$
and instead include the results for the largest grid resolution
for which POMDP solution was possible: higher values could not be handled
within the memory constraints of our test machine.
We anticipate being able to improve this in the future by adapting
more advanced approximation methods for POMDPs~\cite{SPK13}.

\section{Conclusions}

We have proposed novel methods for verification and control of
partially observable probabilistic timed automata,
using a temporal logic for probabilistic, real-time properties and reward measures. We developed techniques based on a digital clocks discretisation and a belief space approximation, then implemented them in a tool and demonstrated their effectiveness on several case studies.

Future directions include more efficient approximation schemes,
zone-based implementations and development of the theory for unobservable clocks. Allowing unobservable clocks, as mentioned previously, will require moving to partially observable stochastic games and restricting the class of strategies.


\vskip9pt \noindent
{\bf Acknowledgments.}
This work was partly supported by the EPSRC grant
``Automated Game-Theoretic Verification of Security Systems'' (EP/K038575/1).
We also grateful acknowledge support from Google Summer of Code 2014.


\bibliographystyle{splncs03}
\bibliography{poptas}

\begin{thebibliography}{10}
\providecommand{\url}[1]{\texttt{#1}}
\providecommand{\urlprefix}{URL }

\bibitem{dA99b}
de~Alfaro, L.: The verification of probabilistic systems under memoryless
  partial-information policies is hard. In: Proc. PROBMIV'99. pp. 19--32 (1999)

\bibitem{AD94}
Alur, R., Dill, D.: A theory of timed automata. Theoretical Computer Science
  126,  183--235 (1994)

\bibitem{BBG08}
Baier, C., Bertrand, N., Gr{\"o}{\ss}er, M.: On decision problems for
  probabilistic {B}\"uchi automata. In: Proc. FOSSACS'08. LNCS, vol. 4962, pp.
  287--301. Springer (2008)

\bibitem{BGL04}
Baier, C., Gr{\"o}{\ss}er, M., Leucker, M., Bollig, B., Ciesinski, F.:
  Controller synthesis for probabilistic systems. In: Proc. TCS'06. pp.
  493--506. Kluwer (2004)

\bibitem{BFH+01b}
{Behrmann, G. et al.}: Minimum-cost reachability for linearly priced timed
  automata. In: Proc. HSCC'01. LNCS, vol. 2034, pp. 147--162. Springer (2001)

\bibitem{BDMP03}
Bouyer, P., D'Souza, D., Madhusudan, P., Petit, A.: Timed control with partial
  observability. In: Proc. CAV'03. LNCS, vol. 2725, pp. 180--192 (2003)

\bibitem{BFLM11}
Bouyer, P., Fahrenberg, U., Larsen, K., Markey, N.: Quantitative analysis of
  real-time systems using priced timed automata. Comm. of the ACM  54(9),
  78--87 (2011)

\bibitem{Cas98}
Cassandra, A.: A survey of {POMDP} applications (1998),
  \url{http://pomdp.org/pomdp/papers/applications.pdf}, presented at the AAAI
  Fall Symposium, 1998

\bibitem{CDL+07}
Cassez, F., David, A., Larsen, K., Lime, D., Raskin, J.F.: Timed control with
  observation based and stuttering invariant strategies. In: Proc. ATVA'07.
  LNCS, vol. 4762, pp. 192--206 (2007)

\bibitem{CCH+11}
Cern{\'{y}}, P., Chatterjee, K., Henzinger, T., Radhakrishna, A., Singh, R.:
  Quantitative synthesis for concurrent programs. In: Proc. CAV'11. LNCS, vol.
  6806, pp. 243--259. Springer (2011)

\bibitem{CCT13}
Chatterjee, K., Chmelik, M., Tracol, M.: What is decidable about partially
  observable {Markov} decision processes with omega-regular objectives. In:
  CSL'13. LIPIcs, vol.~23, pp. 165--180. Schloss Dagstuhl--Leibniz-Zentrum fuer
  Informatik (2013)

\bibitem{CD14}
Chatterjee, K., Doyen, L.: Partial-observation stochastic games: How to win
  when belief fails. ACM Transactions on Computational Logic  15(2) (2014)

\bibitem{Eav84}
Eaves, B.: A course in triangulations for solving equations with deformations.
  Springer (1984)

\bibitem{PF12b}
Finkbeiner, B., Peter, H.: Template-based controller synthesis for timed
  systems. In: Proc. TACAS'12. LNCS, vol. 7214, pp. 392--–406 (2012)

\bibitem{GR12}
Giro, S., Rabe, M.: Verification of partial-information probabilistic systems
  using counterexample-guided refinements. In: Proc. ATVA'12. LNCS, vol. 7561,
  pp. 333--348. Springer (2012)

\bibitem{HMP92}
Henzinger, T., Manna, Z., Pnueli, A.: What good are digital clocks? In: Proc.
  ICALP'92. LNCS, vol. 623, pp. 545--558. Springer (1992)

\bibitem{KMM98}
Kang, M., Moore, A., Moskowitz, I.: Design and assurance strategy for the {NRL}
  pump. Computer  31(4),  56--64 (1998)

\bibitem{KSK76}
Kemeny, J., Snell, J., Knapp, A.: Denumerable {M}arkov Chains (1976)

\bibitem{KNP11}
Kwiatkowska, M., Norman, G., Parker, D.: {PRISM} 4.0: Verification of
  probabilistic real-time systems. In: Proc. CAV'11. LNCS, vol. 6806, pp.
  585--591. Springer (2011)

\bibitem{KNPS06}
Kwiatkowska, M., Norman, G., Parker, D., Sproston, J.: Performance analysis of
  probabilistic timed automata using digital clocks. FMSD  29,  33--78 (2006)

\bibitem{LMT+04}
Lanotte, R., Maggiolo-Schettini, A., Tini, S., Troina, A., Tronci, E.:
  Automatic analysis of the {NRL} pump. ENTCS  99,  245--266 (2004)

\bibitem{LMT05}
Lanotte, R., Maggiolo-Schettini, A., Troina, A.: Automatic analysis of a
  non-repudiation protocol. In: Proc. QAPL'04. ENTCS, vol. 112, pp. 113--129
  (2005)

\bibitem{Lov91}
Lovejoy, W.: Computationally feasible bounds for partially observed {M}arkov
  decision processes. Operations Research  39(1),  162--–175 (1991)

\bibitem{MHC03}
Madani, O., Hanks, S., Condon, A.: On the undecidability of probabilistic
  planning and related stochastic optimization problems. Artif. Intell.
  147(1--2),  5--34 (2003)

\bibitem{MR99}
Markowitch, O., Roggeman, Y.: Probabilistic non-repudiation without trusted
  third party. In: Proc. Workshop on Security in Communication Networks (1999)

\bibitem{NPZ15}
Norman, G., Parker, D., Zou, X.: Verification and control of partially
  observable probabilistic real-time systems. In: Proc. FORMATS'15. LNCS,
  Springer (2015), to appear

\bibitem{NPS13}
Norman, G., Parker, D., Sproston, J.: Model checking for probabilistic timed
  automata. FMSD  43(2),  164--190 (2013)

\bibitem{SPK13}
Shani, G., Pineau, J., Kaplow, R.: A survey of point-based {POMDP} solvers.
  Autonomous Agents and Multi-Agent Systems  27(1),  1--51 (2013)

\bibitem{Yu06}
Yu, H.: Approximate Solution Methods for Partially Observable {M}arkov and
  Semi-{M}arkov Decision Processes. Ph.D. thesis, MIT (2006)

\end{thebibliography}


\newpage
\ifthenelse{\isundefined{\techreport}}{}{%
\appendix
\section{Proof of \thmref{correctness-thm-new}}\label{correctness-appx}

As discussed in \sectref{poptas-sec}, we restrict our attention to POPTAs which reset only non-zero clocks. Given any clock valuations $v,v'$ and distinct sets of clocks $X,Y$ such that $v(x){>}0$ for any $x \in X \cup Y$ we have that $v[X{:=}0]=v'$ implies $v[Y{:=}0] \neq v'$. Therefore, under this restriction, for a time domain $\Tset$ and POPTA $\popta$, if there a exists a transition $(l,v)$ to $(l',v')$ in $\sem{\pta}_\Tset$, then there is a unique set of clocks which are reset when this transition is taken and we define this set as follows.
\begin{definition}\label{reset-def}
For any clock valuations $v,v' \in \tvaluations$ we define the set of clocks $X[v \mapsto v']$ as follows: $X_{[v \mapsto v']} = \{ x \in X \, | \,  v(x){>}0 \wedge v'(x){=}0 \}$.
\end{definition}
Using \defref{reset-def}, the transition function of $\sem{\pta}_\Tset$ is such that, for any $(l,v) \in S$ and $a \in \act$, we have $P((l,v),a) {=} \mu$ if and only if $v \vinz \enb(l,a)$ and for $(l',v') \in S$:
\[
\mu(l',v') = \left\{ \begin{array}{cl}
\probt(l,a)(X_{[v \mapsto v']},l')  & \mbox{if $v[X_{[v \mapsto v']}{:=}0]=v'$} \\
0 & \mbox{otherwise.}
\end{array} \right.
\]
Before we present the proof of \thmref{correctness-thm-new},
we require the concept of a belief PTA.
\begin{definition}\label{belpta-def}
Given a POPTA $\popta {=} \poptatuple$,
the \emph{belief PTA} is given by $\cB(\popta)= \bptatuple$
where
\begin{itemize}
\item
$\dist^{\obsf_\loc} \subseteq \dist(\loc)$ where $\lambda \in \dist^{\obsf_\loc}$ if and only if for any $l,l' \in \loc$ such that $\lambda(l){>}0$ and $\lambda(l'){>}0$ we have $\obsf_\loc(l)=\obsf_\loc(l')$; 
\item
the invariant $\inv^\cB : \dist^{\obsf_\loc} {\ra} \CC{\clocks}$ and enabling conditions $\enb^\cB : \dist^{\obsf_\loc} {\times} \act \ra \CC{\clocks}$ are such that for $\lambda \in \dist^{\obsf_\loc}$ and $a \in \act$ we have $\inv^\cB(\lambda){=}\inv(l)$ and $\enb^\cB(\lambda,a){=}\enb(l,a)$ where $l \in \loc$ is such that $\lambda(l){>}0$;
\item
the probabilistic transition function $\probt^\cB : \dist^{\obsf_\loc} {\times} \act \ra \dist(2^{\clocks} {\times} \dist^{\obsf_\loc})$ is such that for any $\lambda,\lambda' \in \dist^{\obsf_\loc}$, $a \in \act$ and $X \subseteq \clocks$ we have:
\[
\begin{array}{c}
\probt^\cB(\lambda,a)(\lambda',X) = 
\sum\limits_{l \in \loc} \lambda(l) \cdot \left( \sum\limits_{o \in O \wedge \lambda^{a,o,X} = \lambda'} \sum\limits_{l' \in \loc \wedge \obsf_\loc(l')=o} \!\!\!\!\!\!\probt(l,a)(l',X) \right) 
\end{array}
\]
where for any $l' \in \loc$ we have
\[
\lambda^{a,o,X}(l') = \left\{ \begin{array}{cl}
\frac{\sum\nolimits_{l \in \loc} \probt(l,a)(l',X) {\cdot} \lambda(l)}{\sum\nolimits_{l \in \loc} \lambda(l) {\cdot} \left( \sum\nolimits_{l' \in \loc \wedge \obsf_\loc(l')=o} \probt(l,a)(l',X) \right)} & \mbox{if $\obsf_\loc(l'){=}o$} \\
0 & \mbox{otherwise;}
\end{array} \right.
\]
\item 
the reward structure $\ptarew^\cB {=} (\lrew^\cB,\ptaarew^\cB)$ is such that
\[ 
\begin{array}{c}
\lrew^\cB(\lambda) = \sum_{l \in \loc} \lambda(l) {\cdot} \lrew(l) \qquad \mbox{and} \qquad
\ptaarew^\cB(\lambda,a) = \sum_{l \in \loc} \lambda(l) \cdot \ptarew(l,a) \, .
\end{array}
\]
\end{itemize}
\end{definition}
For the above to be well defined we require the conditions of the invariant and observation function given in \defref{popta-def}. For any $\lambda \in \dist^{\obsf_\loc}$ we let $o_\lambda$ be the unique observation such that $\obsf_\loc(l)=o_\lambda$ and $\lambda(l)>0$ for some $l \in \loc$.

Next we introduce the semantics of a PTA which, similarly to \defref{poptasem-def}, is parameterised by a \emph{time domain} $\Tset$.
\begin{definition}[PTA semantics]\label{ptasem-def}
Let $\pta=\ptatuple$ be a probabilistic timed automaton.
The {\em semantics of  $\pta$ with respect to the time domain $\Tset$} is the MDP $\sem{\pta}_\Tset=(S,\sinit,\act \cup \Tset,P,R)$ such that:
\begin{itemize}
\item
$S = \{ (l,v) \in \loc \times \tvaluations \ | \ v \vinz \inv(l)\}$
and
$\sinit = (\linit,\mathbf{0})$;
\item
for any $(l,v) \in S$ and $a \in \act \cup \Tset$, we have $P((l,v),a) = \mu$ if and only if one of the following conditions hold:
\begin{itemize}
\item{(time transitions)}
$a \in \Tset$, $\mu = \dirac{(l,v + a)}$ and $v {+} a \vinz \inv(l)$ for all $0 {\leq} t' {\leq} a$;
\item{(action transition)}
$a \in \act$, $v \vinz \enb(l,a)$ and for $(l',v') \in S$:
\[ \begin{array}{c}
\mu(l',v') = \sum_{X \subseteq \clocks \wedge v' = v[X{:=}0]}  \probt(l,a)(X,l')  
\end{array} \]
\end{itemize}
\item for any $(l,v) \in S$ and $a \in \act \cup \Tset$:
$R((l,v),a) = \left\{ \begin{array}{cl}
\lrew(l){\cdot}a & \mbox{if $a \in \Tset$} \\
\ptaarew(l,a) & \mbox{if $a \in \act$.}
\end{array} \right.$
\end{itemize}
\end{definition}
We now show that for a POPTA $\popta$, the semantics of its belief PTA is equivalent to the belief MDP of the semantics of $\popta$.
\begin{proposition}\label{digital-prop}
For any POPTA $\popta$ which resets only non-zero clocks, time domain $\Tset$ we have that the MDPs $\sem{\cB(\popta)}_\Tset$ and $\cB(\sem{\popta}_\Tset)$ are equivalent.
\end{proposition}
\begin{proof}
Consider any POPTA $\popta=\poptatuple$  which resets only non-zero clocks and time domain $\Tset$. To show the MDPs $\sem{\cB(\popta)}_\Tset$ and $\cB(\sem{\popta}_\Tset)$ are equivalent we will first give a direct correspondence between their state spaces and then show that under this correspondence both the transition and reward functions are equivalent.

Considering the belief MDP $\cB(\sem{\popta}_\Tset)$ and using the fact that $\obsf(l,v)=(\obsf_\loc(l),v)$, for any belief states $b,b'$ and action $a$ we have $P^\cB(b,a)(b')$ equals
\[
\sum\limits_{\substack{(o,v_o) \in O\times \tvaluations \\ b^{a,(o,v_o)}=b'}}\sum\limits_{(l,v) \in S} b(l,v) \cdot \left( \sum\limits_{l' \in \loc \wedge \obsf_\loc(l')=o} P((l,v),a)(l',v_o) \right)  
\]
where for any belief $b$, action $a$, observation $(o,v_o)$ and state $(l',v')$, we have $b^{a,(o,v_o)}(l',v')$ equals
\begin{equation}\label{belief-eqn}
\left\{ \begin{array}{cl}
\frac{\sum_{(l,v) \in S}  P((l,v),a)(l',v') \cdot b(l,v)}{\sum_{(l,v) \in S} b(l,v) \cdot \left( \sum_{l'' \in \loc \wedge \obsf_\loc(l'')=o} P((l,v),a)(l'',v') \right)} & \mbox{if $\obsf_\loc(l'){=}o$ and $v'{=}v_o$} \\
0 & \mbox{otherwise}
\end{array} \right.
\end{equation}
and $R^\cB(b,a) =\sum_{(l,v) \in S} R((l,v),a) \cdot b(l,v)$. Furthermore, by \defref{poptasem-def} and since in $\popta$ only non-zero clocks are reset, if $a \in \act$:
\begin{eqnarray}
P((l,v),a)(l',v') &=& \left\{ \begin{array}{cl}
\probt(l,a)(X_{[v \mapsto v']},l') & \mbox{if $v[X_{[v \mapsto v']}{:=}0]=v'$} \\
0 & \mbox{otherwise}
\end{array} \right.
\label{a-trans-eqn}
\end{eqnarray}
while if  $a \in \Tset$:
\begin{eqnarray}
P((l,v),a)(l',v') &=& \left\{ \begin{array}{cl} 1 & \mbox{if $l'{=}l$ and $v' = v {+} a$} \\
0 & \mbox{otherwise.}
\end{array} \right. \label{t-trans-eqn}
\end{eqnarray}
We see that $b^{a,(o,v_o)}(l',v')$ is zero if $v' \neq v_o$, and therefore we can write the belief as $(\lambda,v_o)$ where $\lambda \in \dist(\loc)$ and $\lambda(l) = b^{a,(o,v_o)}(l,v_o)$ for all $l \in \loc$. In addition, for any $l' \in \loc$, if $\lambda(l'){>}0$, then $\obsf_\loc(l'){=}o$. Since the initial belief $\binit$ can be written as $(\dirac{\linit},\mathbf{0})$ and we assume $\obsf_\loc(\linit) \neq \obsf_\loc(l)$ for any $l \neq \linit \in \loc$, it follows that we can write each belief $b$ of $\cB(\sem{\popta}_\Tset)$ as a tuple $(\lambda,v) \in \dist(\loc) \times \tvaluations$ such that for any $l,l' \in \loc$, if $\lambda(l){>}0$ and $\lambda(l'){>}0$, then $\obsf_\loc(l){=}\obsf_\loc(l')$. Hence, we have shown that the states of $\cB(\sem{\popta}_\Tset)$ are equivalent to the states of $\sem{\cB(\popta)}_\Tset$.

We now use this representation for the states of $\cB(\sem{\popta}_\Tset)$ to rewrite the above equations for the transition and reward function of the belief MDP $\cB(\sem{\popta}_\Tset)$. We have the following two cases to consider. 
\begin{itemize}
\item
For any belief states $(\lambda,v)$ and $(\lambda',v')$ and action $a \in\act$: 
\begin{align*}
\lefteqn{\!\!\!\!\! P^\cB((\lambda,v),a)(\lambda',v')= \; \sum\limits_{\substack{o \in \obs_\loc \\ \lambda^{a,(o,v')}=\lambda'}}\sum\limits_{l \in \loc} \lambda(l) \cdot \left( \sum\limits_{\substack{l' \in \loc \\ \obs(l')=o}}  P((l,v),a)(l',v') \right)} \\
&= \; \sum\limits_{\substack{o \in \obs_\loc \\ \lambda^{a,(o,v')}=\lambda'}}\sum\limits_{l \in \loc} \lambda(l) \cdot \left( \sum\limits_{\substack{l' \in \loc \\ \obs(l')=o}}  \probt(l,a)(X_{[v \mapsto v']},l') \right) & \mbox{by \eqnref{a-trans-eqn}} \\
&= \; \sum\limits_{l \in \loc} \lambda(l) \cdot \left( \sum\limits_{\substack{o \in \obs_\loc \\ \lambda^{a,(o,v')}=\lambda'}} \sum\limits_{\substack{l' \in \loc \\ \obs(l')=o}}  \probt(l,a)(X_{[v \mapsto v']},l') \right) & \mbox{rearranging}
\end{align*}
where for any $l' \in \loc$:
\begin{align*}
\lefteqn{\lambda^{a,(o,v')}(l') 
\; = \; \left\{ \begin{array}{cl}
\frac{\sum_{l \in \loc}  P((l,v),a)(l',v') \cdot \lambda(l)}{\sum_{l \in \loc} \lambda(l) \cdot \left( \sum_{l'' \in \loc \wedge \obsf_\loc(l'')=o} P((l,v),a)(l'',v') \right)} & \mbox{if $\obsf_\loc(l')=o$} \\
0 & \mbox{otherwise}
\end{array} \right.} \\
& = \; \left\{ \begin{array}{cl}
\frac{\sum_{l \in \loc}  \probt(l,a)(X_{[v \mapsto v']},l') \cdot \lambda(l)}{\sum_{l \in \loc} \lambda(l) \cdot \left( \sum_{l'' \in \loc \wedge \obsf_\loc(l'')=o} \probt(l,a)(X_{[v \mapsto v']},l') \right)}  & \mbox{if $\obsf_\loc(l'){=}o$} \\
0 & \mbox{otherwise} 
\end{array} \right. \!\!\!\!\! \!\!\!\!\!& \mbox{by \eqnref{a-trans-eqn}} \\
& = \; \lambda^{a,o,X_{[v \mapsto v']}} & \mbox{by \defref{belpta-def}.}
\end{align*}
Using this result together with \defref{belpta-def} and \defref{ptasem-def} it follows that the transition functions of $\cB(\sem{\popta}_\Tset)$ and $\sem{\cB(\popta)}_\Tset$ are equivalent in the case. Considering the reward functions, we have $R^\cB((\lambda,v),a) = \sum_{l\in \loc} \ptarew_\act(l,a) {\cdot} \lambda(l)$ which, again from \defref{belpta-def} and \defref{ptasem-def}, shows that the reward functions of  $\cB(\sem{\popta}_\Tset)$ and $\sem{\cB(\popta)}_\Tset$ are equivalent.
\item
For any belief states $(\lambda,v)$ and $(\lambda',v')$ and time duration $t \in \Tset$:
\begin{align*}
P^\cB((\lambda,v),t)(\lambda',v') &  = \; \left\{ \begin{array}{cl}
\sum_{l \in \loc} \lambda(l) \cdot P((l,v),a)(l,v')  & \mbox{if $\lambda^{t,(o_\lambda,v')}{=}\lambda'$} \\
0 & \mbox{otherwise} \end{array} \right. 
\end{align*}
where for any $l' \in \loc$:
\begin{align*}
\lambda^{t,(o_\lambda,v')}(l') 
& = 
\left\{ \begin{array}{cl}
\frac{\lambda(l')}{\sum_{l \in \loc} \lambda(l)} & \mbox{if $v'=v {+} t$} \\
0 & \mbox{otherwise}
\end{array} \right. \\
& =  \left\{ \begin{array}{cl}
\lambda(l') & \mbox{if $v'=v {+} t$} \\
0 & \mbox{otherwise}
\end{array} \right. & \mbox{since $\lambda$ is a distribution.}
\end{align*}
Substituting this expression for $\lambda^{t,(o_\lambda,v')}$ into that of $P^\cB((\lambda,v),t)$ we have:
\begin{align*}
\lefteqn{P^\cB((\lambda,v),t)(\lambda',v')} \\
& \lefteqn{= \; \left\{ \begin{array}{cl}
\sum_{l \in \loc} \lambda(l) \cdot \left( \sum_{l' \in \loc} P((l,v),a)(l',v') \right) & \mbox{if $\lambda{=}\lambda'$ and $v'=v {+} t$} \\
0 & \mbox{otherwise} \end{array} \right.} \\
& = \; \left\{ \begin{array}{cl}
\sum_{l \in \loc} \lambda(l)  & \mbox{if $\lambda{=}\lambda'$ and $v'=v {+} t$} \\
0 & \mbox{otherwise} \end{array} \right. & \mbox{by \eqnref{t-trans-eqn}} \\
& = \; \left\{ \begin{array}{cl}
1 & \mbox{if $\lambda{=}\lambda'$ and $v'=v {+} t$} \\
0 & \mbox{otherwise} \end{array} \right. & \mbox{since $\lambda$ is a distribution} 
\end{align*}
which from \defref{belpta-def} and \defref{ptasem-def} shows the transition functions of $\cB(\sem{\popta}_\Tset)$ and $\sem{\cB(\popta)}_\Tset$ are equivalent. For the reward function, we have $R^\cB((\lambda,v),t) = \sum_{l\in \loc} ( \ptarew_\loc(l) {\cdot} t ) {\cdot} \lambda(l)$ and from \defref{belpta-def} and \defref{ptasem-def} this implies that the reward functions of $\cB(\sem{\popta}_\Tset)$ and $\sem{\cB(\popta)}_\Tset$ again are equivalent.
\end{itemize}
Since these are the only cases to consider, both the transition and reward functions of $\cB(\sem{\popta}_\Tset)$ and $\sem{\cB(\popta)}_\Tset$ are equivalent completing the proof.
\qed
\end{proof}
We are now in a position to present the proof of \thmref{correctness-thm-new}.
\begin{proof}[of \thmref{correctness-thm-new}]
Consider any closed diagonal-free POPTA $\popta$ which resets only non-zero clocks and set of observations $O$ of $\popta$. Since the PTA $\cB(\popta)$ is closed and diagonal-free, using results presented in \cite{KNPS06}, we have that:
\begin{equation}\label{thm1-eqn}
\preach{\sem{\cB(\popta)}_\Rset}{opt}{T_O} = \preach{\sem{\cB(\popta)}_\Nset}{opt}{T_O}
\hspace*{0.5em}\mbox{ and }\hspace*{0.6em}
\ereach{\sem{\cB(\popta)}_\Rset}{opt}{T_O} = \ereach{\sem{\cB(\popta)}_\Nset}{opt}{T_O}
\end{equation}
where $opt\in\{\min,\max\}$ and $T_o = \{ (l,v) \, | \, \obsf(l) \in O \}$. Note that, although  \cite{KNPS06} considers only PTAs with finite locations, the corresponding proofs do not require this fact, and hence the results carry over to $\cB(\popta)$ which has an uncountable number of locations. 

Due to the equivalence between a POMDP and its belief MDP we have: 
\begin{equation}\label{thm2-eqn}
\preach{\sem{\popta}_\Tset}{opt}{O} = \preach{\cB(\sem{\popta}_\Tset)}{opt}{T_O}
\hspace*{0.5em}\mbox{ and }\hspace*{0.6em}
\ereach{\sem{\popta}_\Tset}{opt}{O} = \ereach{\cB(\sem{\popta}_\Tset)}{opt}{T_O}. 
\end{equation}
Using \propref{digital-prop} and since $\popta$ resets only non-zero clocks we have $\sem{\cB(\popta)}_\Tset=\cB(\sem{\popta}_\Tset)$. Combining this with \eqnref{thm1-eqn} and \eqnref{thm2-eqn} the theorem follows. \qed
\end{proof}
\section{Construction of the Belief MDP}\label{belief-appx}

For the convenience of the reviewers, below we include, for a given POMDP $\pomdp$, the construction of the corresponding belief MDP $\cB(\pomdp)$.

First, let us suppose we are in a belief state $b$, perform action $a$ and observe $o$. Based on this new information we move to a new belief state $b^{a,o}$ using the observation function of $\pomdp$. Let us now construct this new belief state. First, by construction, we have for any $s' \in S$:
\begin{align*}
 b^{a,o}(s') & = \; \mathbf{Pr}[\, s' \, | \, o,a,b \, ]  \\
& = \frac{\mathbf{Pr}[ \, s' ,o,a,b \, ]}{\mathbf{Pr}[ \, o,a,b \, ]} & \mbox{by definition of conditional probabilities} \\
& = \; \left\{ \begin{array}{cl} 
\dfrac{\mathbf{Pr}[ \, s' ,a,b \, ]}{\mathbf{Pr}[ \, o,a,b \, ]} & \mbox{if $\obsf(s')=o$} \\
 0 & \mbox{otherwise} \end{array} \right.  & \mbox{by definition of $\obsf$.}
\end{align*}
Now considering the numerator in the first case, since the value of $s'$ is dependent on the other values:
\begin{align*}
\mathbf{Pr}[ \, s' ,a,b \, ] & = \;  \mathbf{Pr}[ \, s' \, | \, a,b \, ]\cdot  \mathbf{Pr}[ \, a,b \, ]  \\
& = \;  \mathbf{Pr}[ \, s' \, | \, a,b \, ]\cdot 1  & \mbox{since $b$ and $a$ are fixed} \\
& = \;  \mbox{$\sum_{s \in S}$}  \mathbf{Pr}[ \, s' \, | \, a,s \, ] \cdot b(s)  & \mbox{definition of $b$} \\
& = \;  \mbox{$\sum_{s \in S}$}  P(s,a)(s') \cdot b(s) & \mbox{definition of $P$.}
 \end{align*}
For the denominator since $o$ is dependent on $b$ and $a$ we have:
\begin{align*}
\lefteqn{\mathbf{Pr}[ \, o,a,b \, ] \;= \; \mathbf{Pr}[ \, o \, | \, a,b \, ] \cdot \mathbf{Pr}[ \, a,b \, ]} \\
& = \; \mathbf{Pr}[ \, o \, | \, a,b \, ] \cdot 1 & \mbox{since $b$ and $a$ are fixed} \\
& = \; \mbox{$\sum_{s \in S}$} \mathbf{Pr}[ \, o \, | \, a,s \, ] \cdot b(s) & \mbox{by definition of $b$} \\
& = \; \mbox{$\sum_{s \in S}$} \left( \mbox{$\sum_{s' \in S}$} \mathbf{Pr}[ \, o \, | \, s',a \, ] \cdot P(s,a)(s') \right) \cdot b(s) & \mbox{by definition of $P$} \\
& = \; \mbox{$\sum_{s \in S}$} \left( \mbox{$\sum_{s' \in S \wedge \obsf(s')=o}$} P(s,a)(s') \right) \cdot b(s) & \mbox{by definition of $\obsf$.}
\end{align*}
Combining these results (and rearranging) we have:
\begin{equation*}\label{bao-eqn}
 b^{a,o}(s') \; = \; \left\{ \begin{array}{cl}
\frac{\sum_{s \in S}  P(s,a)(s') \cdot b(s)}{\sum_{s \in S} b(s) \cdot \left( \sum_{s'' \in S \wedge \obsf(s'')=o} P(s,a)(s'') \right)} & \mbox{if $\obsf(s'){=}o$} \\
0 & \mbox{otherwise.}
\end{array} \right.
\end{equation*}
Now using this we can define the probabilistic transition function of the belief MDP $\cB(\pomdp)$. Suppose we are in a belief state $b$ and we perform action $a$. Now the probability we move to belief $b'$ is given by:
\begin{equation*}\label{belieftrans-eqn}
P^\cB_T(b,a)(b') = \mathbf{Pr}[ \, b' \, | \, a,b \, ] \; = \; \mbox{$\sum_{o \in O}$} \mathbf{Pr}[ \, b' \, | \, o,a,b \, ] \cdot \mathbf{Pr}[ \, o \, | \, a,b \, ] \, .
\end{equation*}
The first term in the summation ($\mathbf{Pr}[ \, b' \, | \, o,a,b \, ]$) is the probability of being in belief $b'$ after being in belief $b$, performing action $a$ and observing $o$. Therefore by definition of $b^{a,o}$, this probability will equal 1 if $b'$ equals $b^{a,o}$ and 0 otherwise.
\\ \\
For the second term, as in the derivation of the denominator above, we have:
\begin{align*} \label{oba-eqn}
\mathbf{Pr}[ \, o \, | \, a,b \, ] = \mbox{$\sum_{s \in S}$} b(s) \cdot \left( \mbox{$\sum_{s' \in S \wedge \obsf(s')=o}$} P(s,a)(s') \right) 
\end{align*}
This completes the construction of the transition function of the belief MDP. 
\\ \\
It remains to construct the reward function $R^\cB$ of the belief MDP. In this case, we just have to take the expectation of the reward with respect to the current belief, i.e. for any action $a$ and belief state $b$:
\begin{equation*}\label{beliefreward-eqn}
R^\cB(b,a) = \mbox{$\sum_{s \in S}$} R(s,a) \cdot b(s) \, .
\end{equation*}
The optimal values for the belief MDP equal those for the POMDP, e.g.\ we have:
\[
\preach{\pomdp}{\max}{O} = \preach{\cB(\pomdp)}{\max}{T_O} 
\; \; \mbox{and}
\; \; \ereach{\pomdp}{\max}{O} = \ereach{\cB(\pomdp)}{\max}{T_O}
\]
where $T_O = \{ b \in \dist(S) \, | \, \forall s \in S .\, (b(s){>}0 \ra \obsf(s)\in O) \}$.
}%
\end{document}